\documentclass[american,aps,pra,reprint,superscriptaddress, longbibliography]{revtex4-1}
\usepackage[T1]{fontenc}
\usepackage[latin9]{inputenc}
\usepackage{color}
\usepackage{babel,lipsum}
\usepackage{amsthm,float}
\usepackage{amsmath}
\usepackage{amssymb}
\usepackage{graphicx}
\usepackage[unicode=true,pdfusetitle, bookmarks=true,bookmarksnumbered=false,bookmarksopen=false,  breaklinks=false,pdfborder={0 0 0},backref=false,colorlinks=false] {hyperref}
\hypersetup{ colorlinks,linkcolor=myurlcolor,citecolor=myurlcolor,urlcolor=myurlcolor}
\usepackage{txfonts}%\usepackage{braket}
\usepackage{colortbl}\definecolor{myurlcolor}{rgb}{0,0,0.7}
\usepackage{soul}
\usepackage{qcircuit}

\makeatletter
\@ifundefined{textcolor}{}
{%
 \definecolor{BLACK}{gray}{0}
 \definecolor{WHITE}{gray}{1}
 \definecolor{RED}{rgb}{1,0,0}
 \definecolor{GREEN}{rgb}{0,1,0}
 \definecolor{darkgreen}{rgb}{0,0.8,0}
 \definecolor{BLUE}{rgb}{0,0,1}
 \definecolor{CYAN}{cmyk}{1,0,0,0}
 \definecolor{MAGENTA}{cmyk}{0,1,0,0}
 \definecolor{YELLOW}{cmyk}{0,0,1,0}
 }

\theoremstyle{plain}
\newtheorem{thm}{\protect\theoremname}
\theoremstyle{plain}
\newtheorem{prop}[thm]{\protect\propositionname}
\ifx\proof\undefined
\newenvironment{proof}[1][\protect\proofname]{\par
\normalfont\topsep6\p@\@plus6\p@\relax
\trivlist
\itemindent\parindent
\item[\hskip\labelsep
\scshape
#1]\ignorespaces
}{%
\endtrivlist\@endpefalse
}
\providecommand{\proofname}{Proof}
\fi
\theoremstyle{plain}
\newtheorem{lem}[thm]{\protect\lemmaname}
\newtheorem{theorem}[thm]{\protect\theoremname}
\newtheorem{defi}[thm]{\protect\definitionname}
\newtheorem{corollary}[thm]{\protect\corollaryname}
\newtheorem*{theorem*}{(Theorem}
\newtheorem*{prop*}{(Proposition}
\newtheorem*{lemma*}{(Lemma}
\newtheorem*{defi*}{(Definition}
\newtheorem*{corollary*}{Corollary}

\providecommand{\lemmaname}{Lemma}
\providecommand{\theoremname}{Theorem}
\providecommand{\definitionname}{Definition}
\providecommand{\propositionname}{Proposition}
\providecommand{\corollaryname}{Corollary}

\newcommand{\tr}{{\operatorname{Tr\,}}}

\newcommand{\nc}{\newcommand}
\nc{\rnc}{\renewcommand}

\def\ket#1{| #1 \rangle}
\def\bra#1{\langle  #1 |}
\def\proj#1{| #1 \rangle\!\langle #1 |}
\nc{\avg}[1]{\langle#1\rangle}

\newcommand{\scp}[2]{\langle  #1 |#2 \rangle}
\newcommand{\ketbra}[2]{|#1\rangle \! \langle #2|}

\begin{document}

\title{Quantum Operations in an Information Theory for Fermions}
 
\author{Nicetu Tibau Vidal}
\email{anicet.tibauvidal@physics.ox.ac.uk}
\affiliation{University of Oxford, Clarendon Laboratory, Atomic and Laser Physics Sub-Department. Oxford, Oxfordshire, United Kingdom  }

\author{Mohit Lal Bera}
\affiliation{ICFO -- Institut de Ciencies Fotoniques, The Barcelona Institute of Science and Technology, ES-08860 Castelldefels, Spain}

\author{Arnau Riera}
\affiliation{Institut el Sui, Carrer Sant Ramon de Penyafort, s/n, 08440 Cardedeu (Barcelona), Spain}
\affiliation{ICFO -- Institut de Ciencies Fotoniques, The Barcelona Institute of Science and Technology, ES-08860 Castelldefels, Spain}
\affiliation{Max-Planck-Institut f\"ur Quantenoptik, D-85748 Garching, Germany}

\author{Maciej Lewenstein} 
\affiliation{ICFO -- Institut de Ciencies Fotoniques, The Barcelona Institute of Science and Technology, ES-08860 Castelldefels, Spain}
\affiliation{ICREA, Pg.~Lluis Companys 23, ES-08010 Barcelona, Spain} 

\author{Manabendra Nath Bera}
\email{mnbera@gmail.com}
\affiliation{Department of Physical Sciences, Indian Institute of Science Education and Research (IISER), Mohali, Punjab 140306, India}

\begin{abstract}
A reasonable quantum information theory for fermions must respect the parity super-selection rule to comply with the special theory of relativity and the no-signaling principle. This rule restricts the possibility of any quantum state to have a superposition between even and odd parity fermionic states. It thereby characterizes the set of physically allowed fermionic quantum states. Here we introduce the physically allowed quantum operations, in congruence with the parity super-selection rule, that map the set of allowed fermionic states onto itself. We first introduce unitary and projective measurement operations of the fermionic states. We further extend the formalism to general quantum operations in the forms of Stinespring dilation, operator-sum representation, and axiomatic completely-positive-trace-preserving maps. We explicitly show the equivalence between these three representations of fermionic quantum operations. We discuss the possible implications of our results in characterization of correlations in fermionic systems.
\end{abstract}
\maketitle

\section{Introduction}
Information theory lays one of the foundations of modern science and technologies. In particular, its classical counterpart has revolutionized the realm of information technology, since the time Shannon introduced the classical framework to treat information \cite{Shannon48a, Shannon48b}. The recent developments in quantum technology and understanding of the quantum nature of information have sparked the possibility of information technology that could outperform the classical one \cite{Nielsen00} in presence of quantum resources, such as entanglement \cite{Horodecki09}. The understanding and quantification of quantum information and resources in distinguishable systems, e.g., qudits, rely on the underlying Hilbert space's separability. The global space of composite systems becomes a tensor-product of local Hilbert spaces. However, this assumption falls short when one considers information theory beyond distinguishable systems, as in indistinguishable particle systems, e.g., fermions. Physical systems that resemble qudit-like structures are very limited in nature and generally adhere to quantum mechanics' first quantization. On the other hand, much of the physical systems are understood in terms of indistinguishable particles and quantum fields, using the framework of second quantization. Therefore it is essential to extend information theory to the indistinguishable particles and the quantum fields.  

The inadequacy of the existing framework of quantum information theory (QIT) results from the fact that, for indistinguishable particles, the underlying global Hilbert spaces are either symmetrized (for bosons) or anti-symmetrized (for fermions) part of local (particle or mode) Hilbert spaces. Therefore, Hilbert spaces do not satisfy separability as for qudit cases. To lay a consistent formulation of QIT for indistinguishable particles is essential to understand quantum operations and states, be it local or global, in these ``non-separable'' Hilbert spaces. The latter is associated with entanglement, which has been studied quite extensively in the last few decades. 

For fermions, the literature considers two different approaches to characterize quantum information and resources. Authors differ on whether they can impose separability on the Hilbert spaces from particles or modes' perspectives. In characterizing entanglement,  in \cite{Schliemann01, Schliemann01a, Eckert02, Wiseman03, Ghirardi04, Kraus09, Plastino09, Iemini13, Iemini14, Iemini15, Debarba17, Oszmaniec13, Oszmaniec14, Oszmaniec14a, Sarosi14, Sarosi14a, Majtey16, Kruppa:2020rfa} they assume the \emph{particle} approach, where they consider the indistinguishable particles as subsystems, and the states are the elements of the corresponding local Hilbert spaces. Another approach is based on  \emph{mode} \cite{Zanardi02, Shi03, Friis13, Benatti14, Puspus14, Tosta:2018iqh, Shapourian, Shapourian:2018ozl, Li:2018xon, Szalay, Debarba_2020}, where subsystems are considered to be single-particle modes. One can crudely say that while the notion of particle-based entanglement relies on and exploits quantum mechanics' first-quantization, which strictly conserves particle numbers, the mode-based entanglement is exclusively based on second-quantization of quantum mechanics (using creation and annihilation operators) where particle number may or may not be conserved. It is now clear that the mode-based characterization is more ``reasonable'' in characterizing fermionic entanglement. In this view, one has to impose anti-symmetrization due to the creation and annihilation operators' anti-commutation relations. Moreover, one also has to ascertain that physically allowed fermionic states have to satisfy \emph{parity super-selection rules} \cite{Friis13}. 

In this work, we adhere to the reasonable mode-based approach to fermionic information theory and characterize the physically allowed quantum operations that map the set of fermionic states onto itself. One of our main results proves that all physically allowed unitary operations applicable to fermionic systems must respect the parity super-selection rule to comply with the no-signaling principle. We introduce the general form of reasonable projective operations concerning quantum measurements. The operations are then further extended to the general quantum operations in terms of three equivalent formalisms based on Stinespring dilation, operator-sum-representation (or Kraus representation), and completely-positive-trace-preserving (CPTP) maps. We also discuss the implications of our finding in characterization of correlations present in fermionic systems.

The structure of the paper is as follows. In section \ref{sec:febo} we present a mathematical formalization of the fermionic state space. We precisely define the $\wedge$ product (analogous to $\otimes$), the fermionic Hilbert space, and the parity super-selection rule (SSR). We also give explicit forms to the physically allowed states and observables respecting the parity SSR. The section \ref{sec:quanop} is devoted to the characterization of allowed fermionic quantum operations. In particular, we describe the partial tracing process, the linear operators that preserve the states' parity SSR form, and the unitary and projective operators. We further present the general quantum operations in terms of Stinespring dilation, operator-sum representation, and axiomatic completely-positive-trace-preserving maps.  In section \ref{sec:discussion}, we make a brief discussion on the implications of our findings in fermionic entanglement theory. In particular, we characterize uncorrelated and correlated states, where we also consider the Schmidt decomposition and the purification of states in the context of the parity SSR. A conclusion is made in section \ref{sec:conclusions}.

\section{Fermionic state space}
\label{sec:febo}
Let us consider a set of fermionic modes $I_f$ and denote their corresponding creation and annihilation operators by $f_i^\dagger$ and $f_i$ respectively for any mode $i\in I_f$. These creation and annihilation operators satisfy the standard anti-commutation relations
\begin{align}
\{f_i,f_j\}=0; \ \ \ \{f_i^{\dagger},f_j^{\dagger}\}=0; \ \ \ \{ f_i,f_j^{\dagger}\}=\delta_{ij} \mathbb{I}, \ \ \ \forall i,j\in I_f, \nonumber
\end{align}
with $\{ \hat{A},\hat{B}\}=\hat{A}\hat{B}+\hat{B}\hat{A}$, implying that any multi-mode wave function to be anti-symmetric respect to particle exchange in two different modes. In the case where $|I_f|<+\infty$, there are a finite number of modes and the dimension of the Hilbert space ($\mathcal{H}$)  is also finite. Let us assume that $\mathcal{H}$ is a space for $N$ fermionic modes, i.e. $|I_f|=N$. Then, an order can be chosen in $I_f$, and each mode is referred in respect to its position in $I_f$. The number operator is defined as $\hat{n} = \sum_i f_i^\dagger f_i$. The Hilbert space becomes a direct sum of the subspaces corresponding to different particle number sectors, $ \mathcal{H}=\bigoplus_{n=0}^{|I_f|} \mathcal{H}_{n-p},  \ n=0, \ldots, N$, where $\mathcal{H}_{n-p}$ is the $n$-particle subspace, i.e., any $\ket \psi \in \mathcal{H}_{n-p}$ satisfies $\hat{n} \ket \psi = n \ket \psi$. For a Hilbert (sub-)space $\mathcal{H}_{n-p}$, where $n$-modes are excited, the basis can be reduced to the element $\{f_{i_1}^{\dagger}f_{i_2}^{\dagger}\ldots f_{i_n}^{\dagger} \ket{\Omega}\}$ where $i_1\leq i_2 \leq \ldots \leq i_n$ and $\forall i_j \ \in I_f$. Note that the number of fermionic creation operators applied on the vacuum is $n$. Due to the fact that $f_{i}^{\dagger}f_{j}^{\dagger}=-f_{j}^{\dagger}f_{i}^{\dagger}$, any unordered elements will be equivalent to the ordered element up to a factor $-1$. Further, $f_{j}^{\dagger}f_{j}^{\dagger}=0$ guarantees that the spaces $\mathcal{H}_{n-p}$ for $n>N$ is $\{0\}$. By decomposing the $N$-mode Hilbert space in terms of $ \mathcal{H}_{n-p}$ and ordering the basis in each $\mathcal{H}_{n-p}$, we can fully characterize the space.

\subsection{Wedge product Hilbert space}
\label{subsec:wedge}
At this point, we note that the Hilbert space for indistinguishable particles or subsystems is very different from the distinguishable one. For distinguishable systems, the global Hilbert spaces are the tensor product of the corresponding subsystems' local Hilbert spaces. For example, for two distinguishable systems $A$ and $B$, with the local Hilbert spaces $\mathcal{H}_A$ and $\mathcal{H}_B$ respectively,  the global Hilbert space of $AB$ becomes $\mathcal{H}_{AB}=\mathcal{H}_{A} \otimes \mathcal{H}_{B}$. This also implies that we could form a complete set of basis as $\ket{\psi_A}\otimes \ket{\psi_B} \in \mathcal{H}_{AB}$, in terms of the local basis $\ket{\psi_A}\in \mathcal{H}_A$ and $\ket{\psi_B}\in \mathcal{H}_B$.

In contrast, for a system composed of fermions (or fermionic subsystems), the Hilbert space is anti-symmetrized. The algebraic structure of such a space is ensured by introducing a wedge product ($\wedge$), in place of the tensor product ($\otimes$). Then, for two fermionic systems $A$ and $B$ of $N$ and $M$ modes respectively, the global Hilbert space would be  $\mathcal{H}_{AB}=\mathcal{H}_A \wedge \mathcal{H}_B$. In Appendix \ref{sec:properties}, we expose some essential properties of $\wedge$ product which we use in the following discussions. These properties of the wedge product Hilbert space, defined for fermions, are also compared to those of the tensor product structure for distinguishable particles.
 
By choosing the vacuum state $\ket{\Omega}$ as the identity element of the $\wedge$ product, we find that the following definition of the product is also natural
\begin{align}
\ket{i}\wedge\ket{j}=f_i^{\dagger}\ket{\Omega} \wedge f_j^{\dagger} \ket{\Omega}\equiv f_i^{\dagger}f_j^{\dagger}\ket{\Omega}. \nonumber 
\end{align}

With the anti-commutation relation in mode spaces ($i \in I_{f_A}$ and $j\in I_{f_B}$) it is easy to check that the $\wedge$ product is anti-symmetric, as one would expect for fermions. The extension of the product to the dual spaces is given in the following natural terms
\begin{align}
\left(\ket{i}\wedge\ket{j}\right)^{\dagger}= \left(f_i^{\dagger}f_j^{\dagger}\ket{\Omega}\right)^{\dagger}= \bra{\Omega}f_j f_i \equiv \bra{i}\wedge\bra{j}. \nonumber
\end{align}
This notion can be naturally extended to arbitrary number of particles as the wedge product is associative, i.e., 
\begin{align}
\left(\ket{i_1}\wedge\ldots\wedge\ket{i_N}\right) & \wedge \left(\ket{j_1}\wedge\ldots\wedge\ket{j_M}\right) \nonumber \\
                         & =f_{i_1}^{\dagger}\ldots f_{i_N}^{\dagger} f_{j_1}^{\dagger}\ldots f_{j_M}^{\dagger}\ket{\Omega} \nonumber \\ 
                                                 & =\ket{i_1}\wedge\ldots\wedge\ket{i_N} \wedge\ket{j_1}\wedge\ldots\wedge\ket{j_M} 
\end{align}
This well-defined notion of state vector elements, based on wedge product, form the basis of the Hilbert space. For ease of discussion, we impose an ordering to fix the basis and, then, the Hilbert space of $N$ fermions is expressed with the following proposition. 

\begin{prop}\label{Prop:basis1}
For an $N$-mode fermion system, the set $\mathcal{B}=\{\ket{\Omega},\ket{1},\ket{2},\ket{1}\wedge \ket{2},\ket{3},\ket{1}\wedge\ket{3}, \ket{2}\wedge \ket{3}, \ket{1}\wedge\ket{2}\wedge\ket{3},\ldots,  \ket{1}\wedge\ket{2}\wedge \ldots \wedge \ket{N} \}$ is an orthonormal basis of the Hilbert space $\mathcal{H}$, with dimension $2^N$. 
\end{prop}
The Hilbert space $\mathcal{H}$ can be re-expressed, alternatively, in terms of subspaces $\mathcal{H}_i$ spanned by the basis elements $\{\ket{\Omega},\ket{i}\}$, where $\ket{i}$ is the state when only the $i$th mode is excited. Any arbitrary element $\ket{\psi_i}$, in the Hilbert space belonging to the $i$th mode, can be formed using this basis. Then, in this new representation of an $N$-fermion Hilbert space, every element of $\mathcal{B}$ can be written as $\ket{\psi_1}\wedge\ket{\psi_2}\wedge \ldots \ket{\psi_N}$, where $\ket{\psi_i}\in \mathcal{H}_i$. This enables us to express the global Hilbert space as the wedge product of $\mathcal{H}_i$, and that is
\begin{align}
\mathcal{H}=\bigwedge_{i=1 }^{N} \mathcal{H}_{i}=\mathcal{H}_1 \wedge \ldots \wedge \mathcal{H}_N.
\end{align}
Given the above representation of the Hilbert space $\mathcal{H}$ and its elements, we can study its operator space. A linear operator space $\Gamma (\mathcal{H})$ can be defined with the help of the Cartesian outer product between $\mathcal{H}$ and its dual $\mathcal{H}^*$. More details are given in Appendix \ref{sec:properties}.

Unlike the Hilbert spaces of distinguishable particles and their elements, every element in the fermionic Hilbert space may not represent a physical fermionic state. In fact, every element that represents a physically allowed state has to satisfy a rule, the \emph{parity super-selection} rule, which we discuss next.

\subsection{Parity super-selection rule}
\label{subsec:SSR}
The Hilbert space of fermionic particles is restricted by the algebraic properties due to the anti-commutation relations of creation and annihilation operators and constrained by the \emph{parity super-selection rule} \cite{Wick52}. The super-selection rules (SSRs) are, in general, a set of physical restrictions that are imposed on the states and operations to discard the non-physical cases produced in a physical theory. For the fermionic case, the parity SSR restricts the physical states based on whether a state has an even or odd number of fermions. The parity operator is 
\begin{align}
 \Pi=e^{i \pi \hat{n}},
\end{align}
where $\hat{n}=\sum_i f_i^\dag f_i$ is the number operator. A fermionic state $\ket{\psi^e}$ is even if $ \Pi \ket{\psi^e}=\ket{\psi^e}$, and a state $\ket{\psi^o}$ is said to be odd if  $ \Pi \ket{\psi^o}=-\ket{\psi^o}$. Now the parity SSR can be stated as in the following.  \\

\textbf{Parity SSR} \cite{Wick52}:
\textit{Nature {\bf does not} allow a fermionic quantum state, which is in a coherent superposition between states with even and odd (particle number) parities, i.e. the states $\ket{\phi} = c_e \ket{\psi^e} + c_o \ket{\psi^o}$, where $\{ c_e, \ c_o \} \in \mathbb{C}\backslash\{0\}$, are not physically allowed.}\\

The parity SSR was initially introduced to describe elementary particles in the framework of quantum field theory \cite{Wick52}. It is understood as a consequence of the underlying deep correspondence between quantum mechanics and the (special) theory of relativity \cite{Pauli40, Schwinger51, Peskin95}, through the spin-statistics connection. Beyond its utility in understanding elementary particles, the parity SSR finds important applications to formulate a reasonable quantum information theory for fermions \cite{Friis16}. There is a `partial trace ambiguity' in a fermionic bipartite system if one only relies on fermionic creation and annihilation operators' algebra. For example, it can be seen that there are pure states of a bipartite fermionic system $\ket{\psi}_{AB}$ of fermionic subsystems $A$ and $B$ for which the marginal states $\rho_{A} = \tr_B \left[ \proj{\psi}_{AB} \right]$ and $\rho_{B} = \tr_A \left[ \proj{\psi}_{AB} \right]$ do not share the same spectra. However, a `reasonable' quantum information theory demands that the marginals should possess the same spectra as they share equal information. This serious flaw is resolved by the imposition of parity SSR \cite{Friis16}. It has been shown that any state $\ket{\psi}_{AB}$ which satisfies $\Pi \ket{\psi}_{AB} = \pm \ket{\psi}_{AB} $ results in marginal states with the same spectra. Moreover, in \cite{Johansson16}, it is argued that one does not need to invoke relativity to justify parity SSR for fermions. Rather, the parity SSR can be understood as the consequence of the no-signaling principle, due to the micro-causality constraint \cite{Bogolubov90, Greiner96} on the separable fermionic Hilbert space, in conjugation with the algebra of creation and annihilation operators. It thereby implies that, even in a non-relativistic scenario, fermionic states are constrained by parity SSR. Hence, one cannot ignore SSR while constructing a framework for information theory.        

\subsection{Physical states and observables of fermionic systems}
\label{subsec:physstates}
The physical states of an arbitrary fermionic system must satisfy the parity SSR. Therefore the corresponding Hilbert space $\mathcal{H}_S$, with all physical states satisfying parity SSR, forms a subset, $\mathcal{H}_S \subset \mathcal{H} $. To better represent the set of states that respect parity SSR, we reorder the basis $\mathcal{B}$. The basis $\mathcal{B}$ in an $N$-mode fermionic Hilbert space has $2^N$ elements, and it is easy to see that there are $2^{N-1}$ even elements and $2^{N-1}$ odd elements in it. The new ordering of the elements in the basis set, denoted by $\mathcal{B}'$, is considered in the proposition below.
\begin{defi}\label{Prop:basis2}
 For an $N$-mode fermionic system, the basis set $\mathcal{B}' \equiv \{\mathcal{B}_e,  \mathcal{B}_o \}$ is formed by reordering the elements in $\mathcal{B}$, where the first $2^{N-1}$ elements $\mathcal{B}_e $ are even and last $2^{N-1}$ elements $\mathcal{B}_o $ are odd.
\end{defi}

For example, for a $3$-mode fermionic system, the $ \mathcal{B}_e = \{\ket{\Omega}, \ket{1}\wedge \ket{2}, \ket{1}\wedge \ket{3}, \ket{2}\wedge \ket{3} \}$ and $ \mathcal{B}_o = \{\ket{1}, \ket{2}, \ket{3}, \ket{1}\wedge \ket{2} \wedge \ket{3} \}$. Any arbitrary fermionic pure state $\ket{\psi} \in \mathcal{H}_S$ is either an even or an odd state, and it is expressed as a coherent superposition among the elements either from $\mathcal{B}_e$ or from $\mathcal{B}_o$ respectively. Therefore, an $N$-mode fermionic pure state satisfying parity SSR can be cast in the following form:
\begin{equation}
 \ket{\psi}= \left\{\begin{array}{lr}
   \ket{\psi^e} = \sum_{\vec{s}} \alpha_{\vec{s}} \left(f^{\dagger}_1\right)^{s_1} \ldots \left(f^{\dagger}_N\right)^{s_N} \ket{\Omega}, \ \forall (S \mod 2)=0, \\
   \ket{\psi^o} = \sum_{\vec{s}} \alpha_{\vec{s}} \left(f^{\dagger}_1\right)^{s_1} \ldots \left(f^{\dagger}_N\right)^{s_N}\ket{\Omega}, \ \forall (S \mod 2)=1, 
\end{array}  \right. \nonumber
\end{equation}
where $\vec{s}=\{s_i \}_1^N$,  $s_i=\{0,1\}$, $S=\sum_i s_i$, and $\sum_{\vec{s}} |\alpha_{\vec{s}}|^2=1$. The $\sum_{\vec{s}}$ denotes the sum over all possible values of the vector. The even states $\ket{\psi^e}$ (odd states $\ket{\psi^o}$) are the results of coherent superposition between the elements in $\mathcal{B}_e$ ($\mathcal{B}_o$). By construction, the  $\ket{\psi^e}$ and $\ket{\psi^o}$ are the fermionic states that correspond to even and odd parities, i.e. $\Pi \ket{\psi^e}=\ket{\psi^e}$ and  $\Pi \ket{\psi^o}=-\ket{\psi^o}$ respectively. 

Note, the fermionic operator space $\Gamma_S$ can also be constructed in terms of the elements $\ketbra{\psi}{\varphi}$ where both $\ket{\psi}, \ket{\varphi} \in \mathcal{H}_S$. Using fermionic operators, more general (or mixed) states can be expressed in terms of density matrices. The mixed states are defined as ensembles of pure states. Therefore, the parity SSR imposes restrictions on the matrices that could represent a physical state. A density matrix $\rho$ representing a physically allowed fermionic state, in the basis $\mathcal{B}'$, takes the form
\begin{align}
\rho= \rho_e \oplus \rho_o =\begin{pmatrix}
\rho_e && 0 \\ 0 && \rho_o
\end{pmatrix},
\end{align}
where $\rho_e$ and $\rho_o$  are $2^{N-1}\times 2^{N-1}$ positive semi-definite matrices, given by
\begin{align}
 \rho_e=\sum_i p^e_i \ketbra{\psi^e_i}{\psi^e_i} \ \ \mbox{and} \ \ \rho_e=\sum_i p^o_j \ketbra{\psi^o_j}{\psi^o_j}, 
 \label{eq:rho}
\end{align}
with $0 \leqslant p^e_i, p^o_j \leqslant 1$ and $\sum_i p^e_i + \sum_j p^o_j=1$. By construction, the density matrices are symmetric under parity, i.e. $\Pi \rho \Pi^\dag = \rho$. This also implies that $\Pi \rho_e \Pi^\dag = \rho_e$ and $\Pi \rho_o \Pi^\dag = \rho_o$. The properties of fermionic density matrices are discussed in Appendix \ref{sec:proofs} in detail. The density matrices are positive semi-definite matrices and all positive semi-definite operators form a subset $\mathcal{R}_S \subset \Gamma_S$, where $\rho \in \mathcal{R}_S$.

Unlike distinguishable systems, the physically allowed fermionic observables cannot be represented by any Hermitian operators. Rather, the observables have to respect the parity SSR. It constrains a physically allowed observable $\hat{A} \in \Gamma_S$ to have the form
\begin{align}
 \hat{A}=\sum_i a_i \ketbra{\psi_i}{\psi_i} = \hat{A}_e \oplus \hat{A}_o = \begin{pmatrix}
\hat{A}_e && 0 \\ 0 && \hat{A}_o
\end{pmatrix},
\end{align}
where $a_i \in \mathbb{R}$ and $\ket{\psi_i} \in \mathcal{H}_S$. 

\section{Fermionic Quantum Operations}
\label{sec:quanop}
So far, we have seen how parity SSR imposes conditions on the states and the observables for the fermionic systems. Now we explore how the wedge product structure of Hilbert space and the parity SSR results in a restricted class of allowed quantum operations on the fermionic space that is physical. We start by re-considering the partial tracing operation, initially introduced in \cite{Friis16}. We re-establish the same procedure using a consistency condition (discussed below) and show how the restrictions of fermionic observables and states due to the parity SSR play no role in the notion nor the form of fermionic partial tracing procedure. Then, we move on to classify arbitrary unitary and projection operations that respect the parity SSR structure. Finally, we extend such characterization to general quantum operations.

\subsection{Partial tracing}
\label{subsec:pt}
For the re-derivation of partial tracing operation, let us give a precise definition of a local operator.
\begin{defi}[Local operators]
Consider a global Hilbert space $\mathcal{H}^{AB}_S = \mathcal{H}^A_S \wedge \mathcal{H}^B_S$ of finite fermionic systems $A$ and $B$. An operator $\hat{O}\neq 0$ that acts on $\mathcal{H}_S^{AB}$ is said to be local on $\mathcal{H}^A_S$ if, and only if, it has the form $\hat{O}=\hat{O}_A\wedge \mathbb{I}_B$ with $\hat{O}_A\neq 0$.
\end{defi}

This definition is equivalent to say that all operators local on $\mathcal{H}^A_S$ can be formed by combining creation and annihilation operators of the modes of $\mathcal{H}^A_S$. For more details about this correspondence, see Appendix \ref{sec:properties}. With this notion of local operators, we go on to define partial tracing in a fermionic system with the help of following consistency conditions in order to be able to interpret the reduced density states physically. 

\begin{defi}[Consistency conditions]
Given a global Hilbert space $\mathcal{H}_S^{AB} =\mathcal{H}_S^A \wedge \mathcal{H}_S^B$ of finite-mode fermionic systems $A$ and $B$, the consistency conditions for a partial tracing procedure over $\mathcal{H}_S^B$ of a density matrix $\rho_{AB} \in \mathcal{R}^{AB}_S$ are: that $\tr_{B}(\rho_{AB})=\rho_{A}\in \mathcal{R}^A_S$ is a density matrix and that the following equations are satisfied:
\begin{align}
\tr(\hat{O}_{A} \rho_{AB})=\tr (\hat{O}_{A}\rho_{A}), 
\end{align}
for all $\hat{O}_{A}$ being a local physical observable, thus being an Hermitian operator in $\Gamma^A_S$.
\end{defi}

The consistency conditions give us the physical definition of the reduced density matrix, imposing that the expectation value for local observables in $A$ has to be the same for $\rho_{AB}$ and $\rho_A$. Now, with the consistency conditions above, we put forward Proposition \ref{prop:partialtrace} below that recovers the usual partial tracing procedure prescribed in \cite{Friis16} defined mathematically as $\tr_B(\hat{O}_{AB})=\sum_i \bra{i_B} \hat{O}_{AB} \ket{i_B}$. Thus, we can recover the usual mathematical definition of partial tracing directly from its physical meaning. The proof of Proposition \ref{prop:partialtrace}, along with its equivalence with the procedure considered in \cite{Friis16}, is given in Appendix \ref{sec:ptappendix}.

\begin{prop}[Partial trace]
\label{prop:partialtrace}
For a density operator $\rho \in \mathcal{R}^N_S$, of an $N$-mode fermionic system, the partial tracing over the set of modes $M=\{m_1,\ldots,m_{|M|}\}\subset\{1,\ldots,N\}$ must result in a reduced density operator $\sigma \in \mathcal{R}^{N\setminus M}_S$, given by
\begin{align}
\sigma=\tr_{M}(\rho)=\tr_{m_1}\circ \tr_{m_2} \circ \ldots \circ \tr_{m_{|M|}}(\rho),
\end{align}
 There is a unique partial tracing procedure that satisfies the physically imposed consistency conditions. The operation of partial tracing one mode $m_i$, of an element of $\rho$, is given then by:
\begin{align}
\tr_{m_i} & \left(\left(f_1^{\dagger}\right)^{s_1} \ldots \left(f_{m_i}^{\dagger}\right)^{s_{m_i}}\ldots \left(f_N^{\dagger}\right)^{s_N}\proj{\Omega}f_N^{r_N}\ldots f_{m_i}^{r_{m_i}}\ldots f_1^{r_1} \right)\nonumber \\ 
&= \delta_{s_{m_i} r_{m_i}} (-1)^{k}\left(f_1^{\dagger}\right)^{s_1} \ldots \left(f_N^{\dagger}\right)^{s_N}\proj{\Omega}f_N^{r_N}\ldots f_1^{r_1}, 
\end{align}
with $k=s_{m_i} \sum_{j=m_i}^{N-1} s_{j+1} +r_{m_i} \sum_{k=m_i}^{N-1} r_{k+1}$ and $s_i,r_j \in \{0,1 \}$.
\end{prop}

The partial tracing procedure can be further simplified. Assume that $\ket{a},\ket{b},\ket{c},\ket{d}$ are pure states in $\mathcal{H}_{S}$ of $N$ modes. Say $M\subset \mathcal{N}=\{1,\ldots, N\}$, where $\ket{a}$ and $\ket{b}$ belong to the modes $M$, and $\ket{c}$ and $\ket{d}$ belong to the modes $M^c\equiv N\setminus M$. Then we have
\begin{align}
&\tr_{M}\left(\ketbra{a}{b}\wedge \ketbra{c}{d}\right)=\scp{b}{a}\ketbra{c}{d}, \nonumber \\ 
 &\tr_{M^c}\left(\ketbra{a}{b}\wedge \ketbra{c}{d}\right)=\scp{d}{c}\ketbra{a}{b}. \nonumber 
\end{align}
Note, this now becomes analogous to the partial tracing operation for the distinguishable systems, and the mathematics behind the derivation are outlined in Appendix \ref{sec:proofs}.

\subsection{Linear operators, projectors, and unitaries}
\label{subsec:linear}
Here we consider operators that act on the fermionic Hilbert space $\mathcal{H}_S$. We start with linear operators that respect parity SSR. Any linear operator $\hat{O}$ on the fermionic space assumes the form
\begin{align}
\hat{O}=\sum_{i,j=1}^{2^N} O_{ij} \ketbra{e_i}{e_j}, 
\end{align}
where $\ket{e_i}$ are the ordered elements of the basis set $\mathcal{B}'$. More precisely, the $\ket{e_i}$ and $\ket{e_j}$ have definite parity. The full restrictions on physically allowed linear operators are provided in the following theorem (see Appendix \ref{sec:proofs} for the proof).

\begin{theorem}[Linear operators]
 \label{thm:blockform}
If $\hat{O}$ is a linear operator $\hat{O}: \mathcal{H}^N \rightarrow \mathcal{H}^M $ such that for all $\hat{A}\in \Gamma_S^N$ and $\hat{B}\in \Gamma_S^M$,  $\hat{O} \hat{A}  \hat{O}^\dag \in \Gamma_S^M$ and $\hat{O}^\dagger \hat{B} \hat{O} \in \Gamma_S^N$ then the matrix representation of the operator $\hat{O}$ in the basis $\{{\mathcal{B}'}^N, \ {\mathcal{B}'}^M \}$ assumes one of these two diagonal and anti-diagonal forms:

\begin{gather}
\hat{O}|_{{\mathcal{B}'}^M,{\mathcal{B}'}^N}=\left(\begin{array}{c|c}
O_{++}  & \hat{0} \\ \hline \hat{0}  & O_{--}  
\end{array}\right)  \medspace \medspace \text{or} \medspace \medspace \hat{O}|_{{\mathcal{B}'}^M,{\mathcal{B}'}^N}=\left(\begin{array}{c|c}
\hat{0} & O_{+-}  \\ \hline  O_{-+}  &  \hat{0}
\end{array}\right), \nonumber
\end{gather}
where $O_{++},O_{--},O_{-+} ,O_{+-}  \in M_{2^{M-1}\times 2^{N-1}}(\mathbb{C})$  and $\hat{0}$ is the zero element of $M_{2^{M-1}\times 2^{N-1}}(\mathbb{C})$. Such operators are linear operators that preserve the SSR form of the operators in $\Gamma_S$. 
\end{theorem}

Using this, below we are able to characterize the sharp quantum measurements in terms of projection operations on space $\mathcal{H}_{S}$. $\hat{P} \in \Gamma_S$ is an SSR-respecting projector if it preserves the SSR structure, $\hat{P}^2\ket{\psi}=\hat{P}\ket{\psi}$ and $\hat{P}\ket{\psi}=\hat{P}^{\dagger}\ket{\psi}$, for all $\ket{\psi}\in \mathcal{H}_S$. It can be seen that such conditions imply the following form for the projector $\hat{P}$ when written in the basis $\mathcal{B}'$:
\begin{align}
\hat{P}=P_{ee} \oplus P_{oo} = \left( \begin{array}{c|c}
P_{ee} & \hat{0} \\ \hline 
\hat{0} & P_{oo}
\end{array}\right),
\end{align}
where $P_{ee},P_{oo}\in M_{2^{N-1}\times 2^{N-1}}(\mathbb{C})$ are projectors acting in the even and odd parity sub-spaces respectively. This directly implies that $\hat{P}$ is an SSR projector in an $N$-mode fermionic space $\mathcal{H}$ if, and only if, there exists $\ket{\psi_i} \in \mathcal{H}_{S}$ with $\scp{\psi_i}{\psi_j}=\delta_{ij}$, for $i=1,\ldots,n\leq 2^{N}$, such that
\begin{align}
\hat{P}=\sum_{i=1}^{n} \proj{\psi_i}.
\end{align}

We now turn to characterizing physically allowed unitary operators. A unitary operator $\hat{U}: \mathcal{H}_{S} \rightarrow \mathcal{H}_{S}$ is a linear operator that satisfies $\hat{U} \hat{U}^{\dagger}=\hat{U}^{\dagger} \hat{U}=\hat{\mathbb{I}}_{\mathcal{H}}$. By using Theorem \ref{thm:blockform}, we are able to characterize them.

\begin{theorem}[Unitary]\label{thm:Unitary}
$\hat{U}$ is an SSR-respecting unitary operator acting on $\mathcal{H}_S$ if and only if, the matrix representation of the operator $\hat{U}$ in the basis $\mathcal{B}'$ takes the following form:
\begin{align}
\hat{U}=U_{ee} \oplus U_{oo} = \left(\begin{array}{c|c}
U_{ee}  & \hat{0} \\ \hline \hat{0}  & U_{oo}  
\end{array}\right),
\end{align}
where $U_{ee}$ and $U_{oo}$ are unitary matrices, each in $M_{2^{N-1}\times 2^{N-1}}(\mathbb{C})$, acting on the even and odd sub-spaces respectively, and $\hat{0}$ is the zero element of $M_{2^{N-1}\times 2^{N-1}}(\mathbb{C})$.
\end{theorem}
Theorem  \ref{thm:Unitary} can be proven following two different approaches. The first one relies on the argument that any unitary can be expressed in the form $U=e^{iH}$ where $H$ is a Hermitian operator. Then, the $U$ can only have a non-block-diagonal structure if, and only if, the $H$ also has a non-block-diagonal structure. However, as we have shown earlier, a non-block-diagonal $H$ cannot be a physical fermionic observable (say a Hamiltonian). Therefore, the unitary can never be generated, as such an observable would not be physically allowed. 

In the second approach, the proof is derived by contradiction, and we show that the fermionic anti-diagonal unitaries violate the no-signaling principle. To demonstrate that, let us consider two parties, Alice and Bob. Alice has one qubit $Q$ and a set of $m$ fermionic modes $A$ in her possession. Bob has a set of $n$ fermionic modes $B$ in his possession. The global system is given by the Hilbert space $\mathcal{H}_{Q}\otimes \left(\mathcal{H}_{S}^A \wedge \mathcal{H}_S^B\right)$ corresponding to the qubit and fermionic modes. To start, the initial state is:
\begin{align}
 \rho_i=\ketbra{+}{+}_Q \otimes \rho_{AB},
\end{align}
where the initial qubit state is given by $\ket{+}_Q=\frac{1}{\sqrt{2}}(\ket{0}_Q + \ket{1}_Q)$ and $\rho_{AB}$ is an arbitrary state of the fermionic mode sets $A$ and $B$. Next, we consider the following steps where we process the system with evolutions that are governed by the anti-diagonal unitary operations on the fermionic spaces (see Figure~\ref{fig:circuit}). \\

\begin{figure}
	\centering
	\includegraphics[width=0.950\columnwidth]{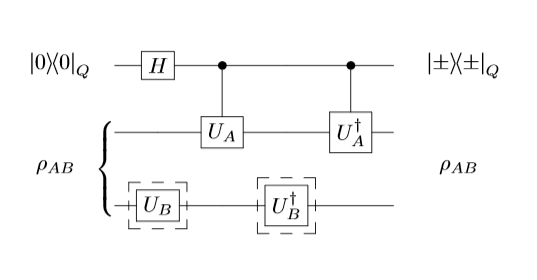}
	\caption{\label{fig:circuit} A scheme that shows the violation of non-signaling by anti-diagonal unitaries. A tripartite system $QAB$ is considered where $Q$ is a qubit system and $A$ and $B$ are fermionic systems. The $H$ represents the Hadamard gate applied on $Q$. Depending on the application of an anti-diagonal fermionic unitaries $U_{B}$ and $U_{B}^\dag$ on $B$, while $QA$ undergoes and anti-diagonal controlled-unitaries $U_{QA}=\ketbra{0}{0}_Q\otimes \mathbb{I}_A + \ketbra{1}{1}_Q \otimes U_A$ and $U_{QA}^\dag$, the local state of $Q$ changes. It implies $B$ can signal $Q$ without having an interaction with it, which violates the no-signaling principle. See text for more details.}
\end{figure}    

\noindent
Step-1: Bob applies a unitary $U_B$ on the modes $B$, where the unitary is anti-diagonal and given by 
\begin{align}
 U_B= \left(\begin{smallmatrix}0 & U_b^{oe} \\ U_b^{eo} & 0 \end{smallmatrix}\right).
\end{align}

\noindent
Step-2: Alice implements a global (control) unitary evolution on both $Q$ and the femionic modes $A$ driven by $U_{QA}=\ketbra{0}{0}_Q\otimes \mathbb{I}_A + \ketbra{1}{1}_Q \otimes U_A$, where 
\begin{align}
 U_A= \left(\begin{smallmatrix}0 & U_a^{oe} \\ U_a^{eo} & 0 \end{smallmatrix}\right),
\end{align}
is an anti-diagonal unitary applied on the modes of $A$.

\noindent
Step-3: Bob applies the unitary $U_B^\dag$ on the modes $B$.

\noindent
Step-4: Alice performs $U_{QA}^\dag$ jointly on the qubit $Q$ and the modes $A$ in her possession. \\

\noindent
After the Step-4, the resultant final state of $QAB$ becomes
\begin{align}
 \rho_f=\ketbra{-}{-}_Q \otimes \rho_{AB},
\end{align}
where the final qubit state is given by $\ket{-}_Q=\frac{1}{\sqrt{2}}(\ket{0}_Q - \ket{1}_Q)$ and the (global) state of the modes $A$ and $B$ does not change. It is clear that if Bob does not apply $U_B$ and $U_B^\dag$ in Step-1 and Step-3 respectively, the overall initial state does not change and the qubit state remains in the state $\ket{+}_Q$. However, Bob's applications of unitaries updates it to $\ket{-}_Q$. Since these qubit states are orthogonal, i.e., $\scp{+}{-}_Q=0$, just by measuring the state of $Q$, Alice can determine whether Bob has applied the unitaries $U_B$ and $U_B^\dag$ or not without having any communication with him. This communication violates the no-signaling principle, and it is exclusively due to the anti-diagonal unitaries. Therefore, the only physically allowed fermionic unitaries are the ones that are block-diagonal. We outline a more detailed version in Appendix \ref{sec:proofs}. 

Therefore, as constrained by the parity SSR, both the fermionic projectors and unitaries must be block-diagonal when the space divides into odd and even subspaces. With this, we move on to characterize general quantum operations in the following result.

\subsection{General quantum operations}
\label{subsec:general}
Here we introduce the general formalism to characterize quantum operations for finite fermionic systems. The formalism considers the algebraic structure of anti-commutation relations of creation and annihilation operators of a fermionic Hilbert space and the parity SSR. There are three different ways to describe general quantum operations or channels for distinguishable quantum systems, and all three approaches are equivalent \cite{Nielsen00}. The first one is physically motivated and describes a general quantum operation as an effect on a system while interacting with an environment. This approach is also known as the \emph{Stinespring dilation} of quantum operations. The second approach, more abstract than the first, is the \emph{operator-sum representation} (also known as \emph{Kraus representation}) of quantum operations. The third method is \emph{axiomatic}, and it is based on the complete positivity and trace preservation of operations, also known as \emph{completely-positive-trace-preserving (CPTP)} operations.

A map often describes an arbitrary quantum operation. Say $\varphi$ maps an arbitrary quantum state of $N$-modes to another state of $M$-modes. Here $N$ and $M$ are not necessarily equal. To characterize general quantum maps, let us define the complete-positivity (CP).

\begin{defi}[CP maps]
A map $\varphi$ is said to be completely positive (CP) if, $\forall \hat{A} \in  \mathcal{R}_{S}^{(N)}$, i.e., $\hat{A} \geq 0$, 
\begin{align}
(\varphi\wedge \mathbb{I}_K)(\hat{A})\geq 0, \ \ \  \forall K\in\mathbb{N}, 
\end{align}
where $\mathbb{I}_K$ is the identity operation acting on the fermionic environment with the Hilbert space $\mathcal{H}_E^K$ of $K$-modes.
\end{defi}
This definition of complete-positivity of a map enables us to present the crucial theorem regarding general quantum channels, below.

\begin{theorem}[General quantum channels] \label{thm:Sinespring-SupO-CPTP}
For an SSR fermionic quantum operation represented by a map $\varphi$, the following statements are equivalent.

\begin{enumerate}
\item (Stinespring dilation.)\label{ax:STPD} There exist a fermionic $K$-mode environment with Hilbert space $\mathcal{H}_{E}^{K}$, a parity SSR-respecting state $\omega \in \mathcal{R}_E^{K}$, and a parity SSR-respecting unitary operator $\hat{U}$ that acts on $\mathcal{H}_{S}^{N}\wedge \mathcal{H}_{E}^{K}$ with $K \geq N$, such that:
\begin{equation}
\varphi(\rho)=\tr_{E}(\hat{U}(\rho\wedge \omega)\hat{U}^{\dagger}), \quad \forall \rho\in \mathcal{R}_{S}^{N}.
\end{equation}

\item (Operator-sum representation.) \label{ax:OSR2} There exists a set of parity SSR-respecting linear operators $E_k$, where $\sum_k E_k^{\dagger}E_k = \mathbb{I}_{N}$, such that:
\begin{equation}
\varphi(\rho)=\sum_k E_k  \rho E_k^{\dagger}, \quad \forall \rho\in \mathcal{R}_{S}^{N}.
\end{equation}

\item (Axiomatic formalism.)\label{ax:CPTP2} $\varphi$ fulfills the following properties:
\begin{itemize}
\item It is trace preserving, i.e. $\tr(\varphi(\rho))= \tr(\rho)$,  $\forall \rho\in \mathcal{R}_{S}^{N}$.
\item Convex-linear, i.e. $\varphi\left(\sum_i p_i\rho_i \right)=\sum_i p_i \varphi(\rho_i)$, $\forall \rho_i \in \mathcal{R}_{S}^{N}$, where $p_i$s are the probabilities, i.e., $0\leq p_i \leq 1$ and  $\sum_i p_i=1$.
\item $\varphi$ is CP.  
\end{itemize}
\end{enumerate}
\end{theorem}
The complete proof of the theorem is outlined in Appendix \ref{sec:prooftheorem}.  The \emph{Stinespring dilation} based formalism of quantum map is physically motivated. There, a fermionic system with state $\rho \in \mathcal{R}_{S}^{N}$,  interacts with an environment in a state $\omega \in \mathcal{R}_{E}^{K}$ through a global unitary operation $U \in \Gamma_{S}^{N+K}$. The quantum operation acting on the system $\mathcal{R}_{S}^{N}$ is then given by the reduced effect on it, which is obtained by partial tracing over environment degrees of freedom. 

The linear operators $E_k$ are also known as \emph{Kraus operators}. Since the $E_k$ operators map SSR-respecting fermionic states onto themselves, they have to fulfill the restrictions imposed by Theorem \ref{thm:blockform}. The condition $ \sum_k E_k^{\dagger}E_k = \mathbb{I}_{2^N}$ guarantees that the operation $\varphi$ is trace preserving, i.e. $\tr(\varphi(\rho))= 1, \ \forall \rho$ with $\tr(\rho)=1$. The case with $ \sum_k E_k^{\dagger}E_k < \mathbb{I}_{2^N}$  corresponds to  incomplete (or selective) quantum operations, where $\tr(\varphi(\rho))< 1$ for $\tr[\rho]=1$. Note, although physical states, observables, projections and unitaries cannot take the anti-diagonal block form of preserving the parity SSR, the Kraus operators assume both block-diagonal and anti-block-diagonal forms. 

\section{Correlations in fermionic systems}
\label{sec:discussion}
With the physically meaningful notions of quantum states and operations, we now discuss some of our results' implications in the reasonable entanglement theory for finite-dimensional fermionic systems. The mode perspective has appeared to be more reasonable to study the correlations present in finite-dimensional fermionic systems. We recover the usual concepts in entanglement theory by exploiting the wedge product's analogies with the tensor product.

\subsection{Uncorrelated states}
\label{subsec:uncorrelated}
One of the main ingredients for various tasks in quantum information theory is inter-system correlations. In the case of composite systems $AB$ with distinguishable subsystems $A$ and $B$, the global Hilbert space is $\mathcal{H}_{AB}=\mathcal{H}_A \otimes \mathcal{H}_B$. The fully uncorrelated states, across the partition $A$ and $B$, are the product states of the form $\rho_A\otimes \rho_B$. Following the discussion presented in \cite{Banuls09} we see that there are three possible ways to define the uncorrelated fermionic states for a bipartite system of fermionic modes with Hilbert space $\mathcal{H}=\mathcal{H}_A \wedge \mathcal{H}_B$.

(i) The first option is to define them as the SSR states $\rho_{AB}$ that satisfy $\tr(\rho_{AB} (\hat{O}_A\wedge \hat{O}_B))=\tr(\rho_A \hat{O}_A)\tr(\rho_B \hat{O}_B)$ for all $\hat{O}_A, \hat{O}_B$ that are local Hermitian operators. Note here that  $\hat{O}_A$, and $\hat{O}_B$ do not need to respect the parity SSR. (ii) The second option is to directly define the uncorrelated SSR states $\rho_{AB}$ as the ones that are SSR product states, so that $\rho_{AB}=\rho_A\wedge \rho_B$. (iii) The third possibility is to define them as the SSR states that fulfil $\tr(\rho_{AB} (\hat{O}_A\wedge \hat{O}_B))=\tr(\rho_A \hat{O}_A)\tr(\rho_B \hat{O}_B)$ for all $\hat{O}_A, \hat{O}_B$ that are local observables in $A$ and $B$ respectively.

Notice that the difference between the definitions (i) and (iii) is the imposition of the SSR condition on the local Hermitian operators $\hat{O}_A, \hat{O}_B$, since we know from Proposition \ref{prop:observables} that Hermitian operators that respect parity SSR represent the fermionic observables. As shown in Appendix \ref{sec:proofs}, the definitions (i) and (ii) are in fact equivalent. However, the third definition is different from the other two, and this is apparent, particularly in the case of mixed states \cite{Banuls09}. The definition (iii) is the physically reasonable one. The property of a composite fermionic system being uncorrelated or correlated should be defined in physical terms, which can be experimentally tested. For this reason, we recast the definition of uncorrelated states for bipartite systems as: 

\begin{defi}[Uncorrelated states \cite{Banuls09}]\label{defi:UnCorr}
Given a bipartite fermionic state $\rho_{AB}$, it is uncorrelated across the partition $A$ and $B$, $\forall \hat{O}_A\in \Gamma_S^A$ and $\forall \hat{O}_B\in \Gamma^B_S $, if
\begin{gather}
\tr(\rho_{AB} (\hat{O}_A\wedge \hat{O}_B))=\tr(\rho_A \hat{O}_A)\tr(\rho_B \hat{O}_B), 
\end{gather}
where $\mathcal{O}_X$ are the local observables in the subspace spanned by the modes on $X=A, B$. 
\end{defi}

There are several distinctions and justifications required to arrive at this result, and that can be traced back to a fundamental difference between fermionic and distinguishable systems: the violation of local tomography principle for fermions \cite{Ariano}. In distinguishable systems, it is known that the local tomography principle is satisfied: given any bipartite state $\rho_{AB}$, the state can be fully reconstructed only by performing local measurements on $A$ and $B$. Nevertheless, as pointed out in \cite{Ariano}, the local tomography principle is not fulfilled in fermionic systems. In fact, there exist states $\rho_{AB} \neq \rho_A \wedge \rho_B$ that have the same expectation values as $\rho_A \wedge \rho_B$ for local observables on $A$ and $B$. 

\subsection{Separable and entangled states}
\label{subsec:entangled}
In quantum entanglement theory, the states created using local-operation-and-classical-communication (LOCC) on uncorrelated states are separable states. With the definition of uncorrelated and physically allowed quantum operations introduced in previous sections, the fermionic separable states shared by parties $A$ and $B$, are given by
\begin{align}
 \rho_{AB}= \sum_i p_i \ \rho_i^{AB} ,
\end{align}
where $0 \leqslant p_i \leqslant 1$, $\sum_i p_i=1$, and the states $\rho_i^{A B} \in \mathcal{R}_S^{AB}$ are uncorrelated states for the bipartition of the sets of modes $A$ and $B$. Obviously the $ \rho_{AB} \in \mathcal{R}_S^{AB}$ is an allowed fermionic state. Note that this definition is in agreement with the one introduced in \cite{Banuls09}. By definition, any bipartite state has non-vanishing correlations across the partition if it is not uncorrelated. The correlations exhibited in separable states can be quantified in terms of classical correlations and quantum discords, such as classical-quantum, quantum-classical, and quantum-quantum correlations \cite{Modi12}.

By definition, an entangled state $\rho_{AB} \in \mathcal{R}_S^{AB}$  shared by two parties $A$ and $B$, is a state that is not separable. Now, one may go on to quantify the amount of entanglement present in a state. In general, it is challenging to characterize and quantify entanglement in a state. However, entanglement in pure states can be characterized easily by using  Schmidt decomposition, a notion that we formalize for fermionic systems in the following lines. 

The Schmidt decomposition for pure bipartite fermionic states in the mode picture is similar to bipartite systems with distinguishable parties, given the theorem below.
 
\begin{theorem}[Schmidt decomposition]\label{thm:schmidt}
Given any bipartite system, any pure SSR fermionic state $\ket{\psi}_{AB} \in \mathcal{H}_{S}^{AB}$, there exist orthonormal basis $\{\ket{i}_A\} \in \mathcal{H}_{S}^A$ and $\{\ket{i}_B\} \in \mathcal{H}_{S}^B$, such that 
\begin{align}
\ket{\psi}_{AB}=\sum_i \sqrt{p_i} \ket{i}_A\wedge \ket{i}_B, 
\end{align}
where $\{p_i\}$ are probabilities.
\end{theorem}
The proof of Theorem \ref{thm:schmidt} is outlined in Appendix \ref{sec:proofs}. The $\{p_i\}$ are called Schmidt coefficients and the number of elements in the set $\{p_i\}$ is called Schmidt number. One may directly see that, a pure fermionic state $\ket{\psi}_{AB}\in \mathcal{H}_{S}^{AB}$ is an uncorrelated state if and only if $\ket{\psi}_{AB}$ possesses Schmidt number $1$.

We can also prove one of the important results in information theory, also analogue to the distinguishable case: the purification of any fermionic SSR quantum state. This result is a corollary of the Schmidt decomposition Theorem in the mode picture \ref{thm:schmidt}. We provide the proof in Appendix \ref{sec:proofs}.
\begin{corollary}[Purification]
	\label{cor:purif}
	If $\rho \in \mathcal{R}_{S}^M$ is an $M$-mode fermionic state, then there exists a fermionic space $\mathcal{H}_E^M$ of $M$ modes and a pure state $\omega \in \mathcal{H}_{S}^M  \wedge \mathcal{H}_E^M$, such that $\tr_E(\omega)=\rho$.
\end{corollary}

\section{Conclusions}
\label{sec:conclusions}

Although the characterization of reasonable fermionic states has been done previously, the identification of reasonable fermionic quantum operations was missing so far. In this work, we study the physically allowed operations that an arbitrary fermionic system may undergo. We show that the operations must satisfy the parity SSR in order to be physically allowed. 

We have started by introducing a wedge-product structure for the Hilbert space of a composite fermionic system that considers the anti-commutation relations of the creation and annihilation operators. Such product arises from the natural separation of the fermionic system into the subsystems of fermionic modes. In addition to that, by applying the parity SSR, we have identified the physical notions of states and observables. Using the framework, we have proven the uniqueness of the partial tracing procedure for fermionic states and characterized the projection and unitary operations. In particular, we have shown that the un-physical unitary operators may lead to violations of the no-signaling principle. We have then extended our studies to characterize general quantum operations in terms of the Stinespring dilation, the operator-sum-representation, and the axiomatic representation based on completely-positive trace-preserving maps. We have shown the equivalence between these three representations. We have also explored the implications of our findings in studying uncorrelated and correlated fermionic states. In particular, we have introduced Schmidt decomposition to characterize entanglement between modes in a fermionic system. This decomposition, in turn, has enabled us to demonstrate the purification of a general fermionic state. 

Our work has drawn a parallel between the quantum information theories (QITs) for distinguishable systems and fermionic (indistinguishable) systems. In particular, we have shown a close resemblance between the quantum states and operations, except that the QIT for distinguishable cases uses the tensor-product structure of the composite Hilbert space and the fermionic cases exploit the wedge-product structure along with a restriction imposed by parity SSR. With this, we see that much of the QIT developed for distinguishable systems can now be translated to fermionic systems. One may even extend the framework developed here to particles following fractional statistics, such as anyons. \\

%\vspace{2pt}

\section*{Acknowledgements}
We thank Dr. Chiara Marletto for fruitful discussions. NTV acknowledges financial support from 'la Caixa' Foundation (ID 100010434, LCF/BQ/EU18/11650048).  ML, MLB, and AR acknowledge supports from ERC AdG NOQIA, Spanish Ministry of Economy and Competitiveness (``Severo Ochoa'' program for Centres of Excellence in R \& D (CEX2019-000910-S), Plan National FIDEUA PID2019-106901GB-I00/10.13039 / 501100011033, FPI), Fundaci\'o Privada Cellex, Fundaci\'o Mir-Puig, and from Generalitat de Catalunya (AGAUR Grant No. 2017 SGR 1341, CERCA program, QuantumCAT U16-011424, co-funded by ERDF Operational Program of Catalonia 2014-2020), MINECO-EU QUANTERA MAQS (funded by State Research Agency (AEI) PCI2019-111828-2 / 10.13039/501100011033), EU Horizon 2020 FET-OPEN OPTOLogic (Grant No 899794), the National Science Centre, Poland-Symfonia Grant No. 2016/20/W/ST4/00314, and the Spanish Ministry MINECO. MNB gratefully acknowledges financial supports from SERB-DST (CRG/2019/002199), Government of India. 

\bibliography{bib_FermiQIT}

\onecolumngrid
%\newpage
\section*{Appendix}
\appendix

Here, we include the different properties of wedge product Hilbert space, proofs of the Lemmas and Theorems, technical details and explanations eluded in the main text.

\section{Wedge product space}
\label{sec:properties}
We define the wedge product ($\wedge$) in Section \ref{sec:febo} to obtain a product notion in the fermionic space, analogous to the tensor product ($\otimes$) in the distinguishable system case. Here we check all the required properties of the so-called wedge product space, to be able to manipulate it consistently.

We have the space $\mathcal{H}$ defined as the Hilbert space spanned by the orthonormal basis 
\begin{align}
 \mathcal{B}_0=\left\{f_{i_1}^{\dagger}\ldots f_{i_k}^{\dagger} \ket{\Omega} \medspace |  \medspace  i_j\in I_N \  \text{s.t.} \medspace  i_{1}<i_2<\ldots<i_{k} \right\},
\end{align}
 
where $k\in\{0,\ldots,N\}$ and $I_N$ is the ordered set of the $N$ modes. Since we desire to define a product that can be viewed as an analogue to the tensor product for the distinguishable case, it is natural to define it from a natural basis of the space. We denote by $\mathcal{H}$ the set of elements that are products (by products of operators we mean wedge product and its compositions unless stated otherwise) of $f_i^{\dagger}$ on the $\ket{\Omega}$ element. 
One can observe that the product, as they are defined above, are also the elements of the basis set $\mathcal{B}_0$. Now this basis can equivalently be expressed, thanks to this product definition and the definition of $\ket{i}=f_i^{\dagger}\ket{\Omega}$, as
\begin{align}\label{eq:basis0}
\mathcal{B}_0=\left\{\ket{i_1}\wedge \ldots\wedge \ket{i_n} |\medspace i_j<i_{j+1}, \medspace 0\leq n \leq N \right\}, \ \text{where} \ i_j\in I_N. 
\end{align}
 
Once these properties are known, we can observe the structure of the basis $\mathcal{B}$ by fixing a number $N$ of modes. We consider the case where we extend our space to the space generated by our $N$ modes and one extra mode, which is labelled by $N+1$. It is clear that the extended basis set $\tilde{\mathcal{B}}$ contains the basis set $\mathcal{B}$. 
The basis set $\tilde{\mathcal{B}}$ is exactly the elements of the basis set of $\mathcal{B}$ by the wedge product of the elements of the set $\{\ket{\Omega},\ket{N+1}\}$. So we obtain that $\tilde{\mathcal{B}}=\mathcal{B} \wedge \{\ket{\Omega},\ket{N+1}\}$. Moreover, since the set $\{\ket{\Omega},\ket{N+1}\}$ it is exactly the canonical basis for the fermionic space of just the $(N+1)$th mode, we have found a decomposition of the basis sets. 
If we denote the basis $\mathcal{B}$ of  $N$-mode space as $\mathcal{B}^N$, and the basis set in $i$th mode as $E_i = \{\ket{\Omega},\ket{i}\} \in \mathcal{H}_i$, we have 
\begin{align}
\mathcal{B}^{N+1}=\mathcal{B}^N \wedge E_{N+1}=\mathcal{B}^{N-1} \wedge E_N \wedge E_{N+1}=\ldots= \mathcal{B}^{1} \wedge E_2\wedge \ldots \wedge E_{N+1} = E_{1} \wedge E_2\wedge \ldots \wedge E_{N+1}.
\end{align}

So, we finally obtain a decomposition of the general basis in a (wedge) product of the most elementary basis set one can think. It consists in the vacuum ($\ket{\Omega}$) and the only excited state ($\ket{i}$) of each mode ($i$th) in the system. An $N$-mode space $\mathcal{H}$ decomposes as the product of the Hilbert spaces $\mathcal{H}_i$ spanned by the orthonormal bases $E_i$, i.e.
\begin{align}
 \mathcal{H}=\bigwedge_{i=1}^N \mathcal{H}_i.
\end{align}

Once we have the natural basis, one may want to explore these bases as the (wedge) product of bases belong to the subspaces. 
Given $\ket{u}\in \mathcal{H}_u$, $ \ket{v} \in \mathcal{H}_v$ for all $k_u,k_v \in \{0,\ldots,N\}, \medspace \text{and} \medspace i_{\eta} \in I_N$, we define

\begin{align}\label{eq:basestruc}
& \ket{u}=\prod_{j_u=1}^{k_u} f_{i_{j_u}}^{\dagger} \ket{\Omega}, \quad  \ket{v}=\prod_{j_v=1}^{k_v} f_{i_{j_v}}^{\dagger} \ket{\Omega},   \qquad 
& \ket{u} \wedge \ket{v} := \prod_{j_u=1}^{k_u} f_{i_{j_u}}^{\dagger}  \prod_{j_v=1}^{k_v} f_{i_{j_v}}^{\dagger} \ket{\Omega} \in \mathcal{H}_u \wedge \mathcal{H}_v.  
\end{align}
It may be worth commenting that the product is closed in the set of elements that are products of the creation operators acting on the vacuum state $\ket{\Omega}$. In $\mathcal{H}$,  the product inherits the property of being associative from the associativity of the composition of operators. The element $\ket{\Omega}$, corresponds to the vacuum state, is the neutral element of the product, fulfilling that $\forall \ket{u}\in \mathcal{H}$, $\ket{u}\wedge \ket{\Omega}=\ket{\Omega}\wedge \ket{u}= \ket{u}$. It can also be seen that for any pair $\ket{u},\ket{v}\in\mathcal{H}_{u,v}$, $\ket{u}\wedge \ket{v} =\pm \ket{v}\wedge \ket{u}$, where the phase depends on the number of creation operators on each term. If the number of creation operators are even (or both are odd) for both the terms, then $\ket{u}\wedge \ket{v} = \ket{v}\wedge \ket{u}$, otherwise $\ket{u}\wedge \ket{v} = - \ket{v}\wedge \ket{u}$. This follows from the definition of the product in terms of the composition of the creation operators, and then applying the anti-commutation relations for the fermionic case. Using the same argument, the product between two elements $\ket{u},\ket{v} \in \mathcal{H}$ that share same creation operator, is $\ket{u}\wedge \ket{v}=0$. Beyond the bases, any two elements  $ \ket{\psi_u} \in \mathcal{H}_u$ and $ \ket{\psi_v} \in \mathcal{H}_v$, the $\ket{\psi_u}\wedge \ket{\psi_v} \in \mathcal{H}_u \wedge \mathcal{H}_v$.

One of the keys of this formalism is that, although we have fixed an order to work with, any permutation on this order of modes would lead to the same result; e.g. if we consider the product $E_2\wedge E_1$ is also a valid basis of the same space as $E_1\wedge E_2$, moreover it is the same basis with some of its elements multiplied by $-1$. 

From now on, when we refer to the product $\ket{\psi_u}\wedge \ket{\psi_v}$, we assume that they live in disjoint decompositions of spaces. With this notion of $\mathcal{H}$, we now consider its dual space.

The dual space of $\mathcal{H}$ is the space where the elements $\bra{\psi}$ are, elements such that the scalar product defined in $\mathcal{H}$ can be seen as $(\ket{\varphi},\ket{\psi})=\scp{\varphi}{\psi}$. We now want to define how the $\wedge$ product acts on this space in order to have consistency with this condition. If we consider the simplest case, we find that the dual element of $\ket{\Omega}$ is given by $\ket{\Omega}^{\dagger}=\bra{\Omega}$, so that $\scp{\Omega}{\Omega}=1$. Then it is natural to see that for an element $\ket{i}=f_i^{\dagger}\ket{\Omega}$, the corresponding dual element is $\bra{\Omega}f_i$. From the commutation relation it follows that $\bra{\Omega}f_i f_i^{\dagger}\ket{\Omega}=1$, and it makes sense of the $\dagger$ notion of the dual space. Now, in order to see how the product $\wedge$ is defined, we consider the element $\ket{i}\wedge \ket{j}=f_i^{\dagger}f_j^{\dagger}\ket{\Omega}$, and we observe that its dual element is $\bra{\Omega}f_j f_i$ and not $\bra{\Omega} f_i f_j$, because the second gives that the norm of the ket is -1, and the first 1. So we obtained that $\bra{\Omega} f_i f_j=(\ket{i}\wedge\ket{j})^{\dagger}$. Now, since we desire to have similar properties to the tensor product, we would like to have that: $(\ket{i}\wedge\ket{j})^{\dagger}=\bra{i}\wedge\bra{j}$. So, imposing this condition we find that it can be defined a consistent operation for the $\wedge$ product on the dual space by setting: $\bra{i}\wedge\bra{j}:=\bra{\Omega}f_j f_i$. 

Then, performing the same generalization as in the case for the $\mathcal{H}$ space, the duals of the basis in Eq.~\ref{eq:basestruc} are defined as
\begin{align}
& \bra{u}=\bra{\Omega}\prod_{j_u=k_u}^{1} f_{i_{j_u}}, \quad \bra{v}=\bra{\Omega}\prod_{j_v=k_v}^{1} f_{i_{j_v}}, \qquad 
& \bra{u}\wedge \bra{v}:=\bra{\Omega} \prod_{j_v=k_v}^{1} f_{i_{j_v}} \prod_{j_u=k_u}^{1} f_{i_{j_u}}.
\end{align}
Then, by following the same technique as for the $\mathcal{H}$, we have
\begin{align}
\mathcal{H}^*=\mathcal{H}_{1}^* \wedge \ldots \wedge \mathcal{H}_N^* 
\end{align}
where $\mathcal{H}_{i}^*=\{\bra{\Omega},\bra{i}\}$, that matches exactly with the dual set of $\mathcal{H}_i$. From here, it is easy to see that the wedge product behaves well in general with the dual operation. So, for  $\ket{\psi_u} \in \mathcal{H}_u$ and $\ket{\varphi_v}\in\mathcal{H}_v$, then 
\begin{equation}
\left(\ket{\psi_u}\wedge \ket{\varphi_v}\right)^{\dagger}=\bra{\psi_u}\wedge\bra{\varphi_v} \in \mathcal{H}_u^* \wedge \mathcal{H}_v^*.
\end{equation}
From the associativity of the $\wedge$ product on both spaces, we further have
\begin{align}
\left(\ket{\psi_1}\wedge\ldots\wedge\ket{\psi_n}\right)^{\dagger}=\bra{\psi_1}\wedge \ldots \wedge \bra{\psi_n}. 
\end{align} 
With the dual elements, we can easily cast a scalar products. For two basis elements it is given as in the following.
\begin{lem}
\label{lema:scprod}
The scalar product between the elements $f_{i_1}^{\dagger}\ldots f_{i_k}^{\dagger}\ket{\Omega}$ and $f_{j_1}^{\dagger}\ldots f_{j_m}^{\dagger}\ket{\Omega}$ is given by:
\begin{gather}
\bra{\Omega}f_{i_k}\ldots f_{i_1}f_{j_1}^{\dagger}\ldots f_{j_m}^{\dagger}\ket{\Omega}=\begin{cases}
0 & \text{If } k\neq m \\
\det{W_k} & k=m
\end{cases} \qquad \text{where } W_k=\begin{pmatrix}
\delta_{i_1,j_1} & \ldots & \delta_{i_1,j_k} \\
\vdots & \ddots & \vdots\\
\delta_{i_k,j_1} & \ldots & \delta_{i_k,j_k}
\end{pmatrix} 
\end{gather}
\end{lem}
\begin{proof}
First, it has to be seen that if $k\neq m$ then the value is 0. If $k\neq m$ then or $m>k$ or $k>m$.\\

In the first case, by the pigeon hole principle there will be some $l_0 \in \{1,...,m\} $ such that $j_{l_0}\neq i_s \forall s\in\{1,...,k\}$, therefore $\delta_{j_{l_0} s}=0 \forall s\in \{1,...,k\}$. That is why the element $f^{\dagger}_{j_{l_0}}$ can be commuted (with the corresponding sign, towards $\bra{\Omega}$, and annihilate it :\begin{equation}
\bra{\Omega}f_{i_k}\ldots f_{i_1}f_{j_1}^{\dagger}\ldots f_{j_m}^{\dagger}\ket{\Omega}=(-1)^{k+j_{l_0}-1}\bra{\Omega}f^{\dagger}_{j_{l_0}} f_{i_k}\ldots f_{i_1}f_{j_1}^{\dagger}\ldots f_{j_m}^{\dagger}\ket{\Omega}=0
\end{equation}
Similarly this procedure can be repeated in the second case commuting the corresponding $f$ element towards the element $\ket{\Omega}$ and annihilate it as well. So just the case $k=m$ is left to see. The exposed relation will be proved by induction. For $k=1$:
\begin{equation}
\bra{\Omega}f_{i_1}f^{\dagger}_{j_1}\ket{\Omega}=\delta_{i_1 j_1}- \bra{\Omega}f^{\dagger}_{j_1}f_{i_1}\ket{\Omega}=\delta_{i_1 j_1}= \det(\delta_{i_1 j_1})
\end{equation}
For $k=2$ the proof follows by:
\begin{gather}
\bra{\Omega}f_{i_2}f_{i_1}f^{\dagger}_{j_1}f^{\dagger}_{j_2}\ket{\Omega}=\delta_{i_1 j_1}\bra{\Omega}f_{i_2}f^{\dagger}_{j_2}\ket{\Omega}-\bra{\Omega}f_{i_2}f^{\dagger}_{j_1}f_{i_1}f^{\dagger}_{j_2}\ket{\Omega}=\nonumber \\ =\delta_{i_1 j_1}\delta_{i_2 j_2}-\delta_{i_2 j_1}\bra{\Omega}f_{i_1}f^{\dagger}_{j_2}\ket{\Omega}+\bra{\Omega}f^{\dagger}_{j_1}f_{i_2}f_{i_1}f^{\dagger}_{j_2}\ket{\Omega}=\delta_{i_1 j_1}\delta_{i_2 j_2}-\delta_{i_2 j_1}\delta_{i_1 j_2}=\det \begin{pmatrix} \delta_{i_1 j_1} & \delta_{i_1 j_2}\\ \delta_{i_2 j_1} & \delta_{i_2 j_2}
\end{pmatrix}
\end{gather}
So now lets choose a fixed integer $n_0>1$, and lets assume that $\forall n\leq n_0, n\in \mathbb{N} $ the relation is true. If the relation is proved for $n_0+1$ then the proof is done. So lets consider this case: \begin{gather}
\bra{\Omega}f_{i_{n_0+1}}\ldots f_{i_1}f_{j_1}^{\dagger}\ldots f_{j_{n_0+1}}^{\dagger}\ket{\Omega}=\delta_{i_1 j_1} \bra{\Omega}f_{i_{n_0+1}}\ldots f_{i_2}f_{j_2}^{\dagger}\ldots f_{j_{n_0+1}}^{\dagger}\ket{\Omega}-\bra{\Omega}f_{i_{n_0+1}}\ldots f_{i_2}f_{j_1}^{\dagger}f_{i_1}f_{j_2}^{\dagger}\ldots f_{j_{n_0+1}}^{\dagger}\ket{\Omega}=\nonumber \\ =\delta_{i_1 j_1}\det(\tilde{A}_{n_0+1}^{1,1})-\delta_{i_2 j_1} \bra{\Omega}f_{i_{n_0+1}}\ldots f_{i_3}f_{i_1}f_{j_2}^{\dagger}\ldots f_{j_{n_0+1}}^{\dagger}\ket{\Omega} +\bra{\Omega}f_{i_{n_0+1}}\ldots f_{i_3}f_{j_1}^{\dagger}f_{i_2}f_{i_1}f_{j_2}^{\dagger}\ldots f_{j_{n_0+1}}^{\dagger}\ket{\Omega}=\nonumber \\ = \delta_{i_1 j_1}\det(\tilde{A}_{n_0+1}^{1,1})-\delta_{i_2 j_1} \det(\tilde{A}_{n_0+1}^{2,1}) +\delta_{i_3 j_1}\bra{\Omega}f_{i_{n_0+1}}\ldots f_{i_4}f_{i_2}f_{i_1}f_{j_2}^{\dagger}\ldots f_{j_{n_0+1}}^{\dagger}\ket{\Omega}+\nonumber \\ -\bra{\Omega}f_{i_{n_0+1}}\ldots f_{i_4}f_{j_1}^{\dagger}f_{i_3}f_{i_2}f_{i_1}f_{j_2}^{\dagger}\ldots f_{j_{n_0+1}}^{\dagger}\ket{\Omega}=\ldots = \sum_{k=1}^{n_0+1} (-1)^{k+1}\delta_{i_k j_1} \det(\tilde{A}_{n_0+1}^{k,1})+\nonumber \\ + \bra{\Omega}f_{j_1}^{\dagger}f_{i_{n_0+1}}\ldots f_{i_1}f_{j_2}^{\dagger}\ldots f_{j_{n_0+1}}^{\dagger}\ket{\Omega} = \sum_{k=1}^{n_0+1} (-1)^{k+1}\delta_{i_k j_1} \det(\tilde{A}_{n_0+1}^{k,1})=\det(A_{n_0+1})
\end{gather}\end{proof}

With the structure of the dual space of $\mathcal{H}$, we are now ready to the extend of the wedge product structure to the operator space of $\mathcal{H}$, that is $\Gamma$. The operator space is constructed with the Cartesian product between two basis of $\mathcal{H}$ and $\mathcal{H}^*$ and the conventional linear extension of this elements. So, we can say, without loss of generality, that any linear operator on $\mathcal{H}$, a fermionic space of N-modes, takes the form $\hat{A}=\sum_{r=1}^{2^N}\sum_{s=1}^{2^N} \alpha_r \beta_s \ketbra{u_r}{u_s}$, where $\{ \ket{u_r}, \ket{u_s} \} \in \mathcal{B}_0^N$.

Focusing on the action of the $\wedge$ product lets consider a bipartition of the modes $A|B$. From now on, we consider the decomposition of any vector or operator in terms of ordered actions of creation and annihilation operators acting on the $\ket{\Omega}$ or the $\bra{\Omega}$ elements. In case that we have the sum of many of those terms, they will be distinct one from each other. From all these considerations taken into account, we present the following results.
\begin{lem}\label{lema:distribu}
If $\hat{C},\hat{E}$ are local operators on $A$, and $\hat{D},\hat{F}$ are local operators on $B$, then $(\hat{C}\wedge \hat{D})(\hat{E}\wedge \hat{F})=(\hat{C}\hat{E}\wedge \hat{D}\hat{F})$.
\end{lem}

\begin{proof}
Since the products fulfill the distributive and linear properties both in $\mathcal{H}$ and $\mathcal{H}^*$, it is enough to proof it for just the matrix elements, given by 
\begin{gather}
\hat{C}=c f_{\lambda_1}^{\dagger}\ldots f_{\lambda_l}^{\dagger}\proj{\Omega}f_{\mu_m}\ldots f_{\mu_1} \medspace \medspace  \hat{D}=d f_{\nu_1}^{\dagger}\ldots f_{\nu_n}^{\dagger}\proj{\Omega}f_{\pi_p}\ldots f_{\pi_1} \medspace \medspace \hat{E}=e f_{\rho_1}^{\dagger}\ldots f_{\rho_r}^{\dagger}\proj{\Omega}f_{\sigma_s}\ldots f_{\sigma_1} \medspace \medspace \hat{F}=f f_{\tau_1}^{\dagger}\ldots f_{\tau_t}^{\dagger}\proj{\Omega}f_{\xi_x}\ldots f_{\xi_1} \nonumber \\
(\hat{C}\wedge \hat{D})(\hat{E}\wedge \hat{F})=cdef( f_{\lambda_1}^{\dagger}\ldots f_{\lambda_l}^{\dagger} f_{\nu_1}^{\dagger}\ldots f_{\nu_n}^{\dagger}\proj{\Omega} f_{\pi_p}\ldots f_{\pi_1} f_{\mu_m}\ldots f_{\mu_1}) (f_{\rho_1}^{\dagger}\ldots f_{\rho_r}^{\dagger} f_{\tau_1}^{\dagger}\ldots f_{\tau_t}^{\dagger}\proj{\Omega} f_{\xi_x}\ldots f_{\xi_1} f_{\sigma_s}\ldots f_{\sigma_1})=\nonumber \\ =cdef W f_{\lambda_1}^{\dagger}\ldots f_{\lambda_l}^{\dagger} f_{\nu_1}^{\dagger}\ldots f_{\nu_n}^{\dagger}\proj{\Omega}  f_{\xi_x}\ldots f_{\xi_1} f_{\sigma_s}\ldots f_{\sigma_1}=\frac{W}{W_1 W_2}\cdot CE\wedge DF 
\end{gather}
where $W=\bra{\Omega}f_{\pi_p}\ldots f_{\pi_1} f_{\mu_m}\ldots f_{\mu_1}f_{\rho_1}^{\dagger}\ldots f_{\rho_r}^{\dagger} f_{\tau_1}^{\dagger}\ldots f_{\tau_t}^{\dagger}\ket{\Omega}$, $W_1=\bra{\Omega} f_{\mu_m}\ldots f_{\mu_1}f_{\rho_1}^{\dagger}\ldots f_{\rho_r}^{\dagger} \ket{\Omega}$, and $W_2=\bra{\Omega}f_{\pi_p}\ldots f_{\pi_1}  f_{\tau_1}^{\dagger}\ldots f_{\tau_t}^{\dagger}\ket{\Omega}$. The proof is complete if we can show that 
 $W=W_1 W_2 $. For $t\neq p$ or $r\neq m$, we have $W=W_1 W_2=0$. But if they are equal, $t= p$ and $r= m$, then 
\begin{gather}
W=\det \begin{pmatrix}
\delta_{\mu_1 ,\rho_1} & \ldots & \delta_{\mu_1, \rho_r} & \delta_{\mu_1, \tau_1} & \ldots & \delta_{\mu_1, \tau_t} \\
\vdots & \ddots & \vdots & \vdots & \ddots & \vdots \\
\delta_{\mu_m ,\rho_1} & \ldots & \delta_{\mu_m, \rho_r} & \delta_{\mu_m, \tau_1} & \ldots & \delta_{\mu_m, \tau_t} \\ 
\delta_{\pi_1 ,\rho_1} & \ldots & \delta_{\pi_1, \rho_r} & \delta_{\pi_1, \tau_1} & \ldots & \delta_{\pi_1, \tau_t} \\
\vdots & \ddots & \vdots & \vdots & \ddots & \vdots \\
\delta_{\pi_p ,\rho_1} & \ldots & \delta_{\pi_p, \rho_r} & \delta_{\pi_p, \tau_1} & \ldots & \delta_{\pi_p, \tau_t}
\end{pmatrix} = \det \begin{pmatrix}
\delta_{\mu_1 ,\rho_1} & \ldots & \delta_{\mu_1, \rho_r} & 0 & \ldots & 0\\
\vdots & \ddots & \vdots & \vdots & \ddots & \vdots \\
\delta_{\mu_m ,\rho_1} & \ldots & \delta_{\mu_m, \rho_r} & 0 & \ldots & 0 \\ 
0 & \ldots & 0 & \delta_{\pi_1, \tau_1} & \ldots & \delta_{\pi_1, \tau_t} \\
\vdots & \ddots & \vdots & \vdots & \ddots & \vdots \\
0 & \ldots & 0 & \delta_{\pi_p, \tau_1} & \ldots & \delta_{\pi_p, \tau_t}\end{pmatrix}=\nonumber \\ = \det\begin{pmatrix}
\delta_{\mu_1,\rho_1} & \ldots & \delta_{\mu_1, \rho_r}\\ \vdots & \ddots & \vdots \\ \delta_{\mu_m,\rho_1} & \ldots & \delta_{\mu_m, \rho_r}
\end{pmatrix} \det\begin{pmatrix}
\delta_{\pi_1,\tau_1} & \ldots & \delta_{\pi_1, \tau_t}\\ \vdots & \ddots & \vdots \\ \delta_{\pi_p,\tau_1} & \ldots & \delta_{\pi_p, \tau_t}
\end{pmatrix}=W_1 W_2 
\end{gather}
Where we use the Lemma \ref{lema:scprod} in the above manipulation.
\end{proof}

\begin{lem}
If $\hat{P_A}$ and $\hat{P_B}$ are hermitian operators, then $\hat{P_A}\wedge \hat{P_B}$ is also a Hermitian operator.
\label{lema:lho}
\end{lem}
\begin{proof}
Lets write down $\hat{P_A}$ and $\hat{P_B}$ as
\begin{gather}
\hat{P_A}=\sum_{\lambda_1,\ldots, \lambda_l,  \mu_1,\ldots,\mu_m} \eta_{\lambda_1,\ldots, \lambda_l  \mu_1,\ldots,\mu_m} f_{\lambda_1}^{\dagger}\ldots f_{\lambda_l}^{\dagger} \proj{\Omega} f_{\mu_m}\ldots f_{\mu_1}, \quad \hat{P_B}=\sum_{\nu_1,\ldots, \nu_n,  \sigma_1,\ldots,\sigma_s} \beta_{\nu_1,\ldots, \nu_n  \sigma_1,\ldots,\sigma_s} f_{\nu_1}^{\dagger}\ldots f_{\nu_n}^{\dagger} \proj{\Omega} f_{\sigma_s}\ldots f_{\sigma_1}.
\end{gather}
Since, $\hat{P_A}$ and $\hat{P_B}$ are Hermitian the coefficients have to fulfill the following relationship 
\begin{equation}
\eta_{\lambda_1,\ldots, \lambda_l,  \mu_1,\ldots,\mu_m}=\eta^{*}_{\mu_1,\ldots, \mu_m,  \lambda_1,\ldots,\lambda_l}, \ \ \mbox{and} \ \   \beta_{\nu_1,\ldots, \nu_n,  \sigma_1,\ldots,\sigma_s}=\beta^{*}_{\sigma_1,\ldots, \sigma_s,  \nu_1,\ldots,\nu_n}.
\end{equation}
Now, with the behavior of the wedge product on the $\mathcal{H}$ and $\mathcal{H}^*$ spaces, the $\hat{P_A}\wedge \hat{P_B}$ can be expressed as 
\begin{gather}
\hat{P_A}\wedge \hat{P_B} =\sum_{\lambda_1,\ldots, \lambda_l , \mu_1,\ldots,\mu_m, \nu_1,\ldots, \nu_n  \sigma_1,\ldots,\sigma_s} \eta_{\lambda_1,\ldots, \lambda_l  \mu_1,\ldots,\mu_m} \beta_{\nu_1,\ldots, \nu_n  \sigma_1,\ldots,\sigma_s}   f_{\lambda_1}^{\dagger}\ldots f_{\lambda_l}^{\dagger}  f_{\nu_1}^{\dagger}\ldots f_{\nu_n}^{\dagger} \proj{\Omega}  f_{\sigma_s}\ldots f_{\sigma_1} f_{\mu_m}\ldots f_{\mu_1}.
\end{gather}
It will be hermitian if and only if $\alpha_{\lambda_1,\lambda_l,\nu_1,\nu_n; \mu_1,\mu_m,\sigma_1,\sigma_s} = \alpha^{*}_{\lambda_1,\lambda_l,\nu_1,\nu_n; \mu_1,\mu_m,\sigma_1,\sigma_s} $, where $\alpha_{\lambda_1,\lambda_l,\nu_1,\nu_n; \mu_1,\mu_m,\sigma_1,\sigma_s} = \eta_{\lambda_1,\ldots, \lambda_l,  \mu_1,\ldots,\mu_m} \beta_{\nu_1,\ldots, \nu_n,  \sigma_1,\ldots,\sigma_s}  $ and $\alpha^{*}_{\lambda_1,\lambda_l,\nu_1,\nu_n; \mu_1,\mu_m,\sigma_1,\sigma_s} = \eta^{*}_{\lambda_1,\ldots, \lambda_l,  \mu_1,\ldots,\mu_m} \beta^{*}_{\nu_1,\ldots, \nu_n,  \sigma_1,\ldots,\sigma_s} $. Using the impositions $\eta=\eta^*$ and $\beta=\beta^*$, we see that indeed $\alpha_{\lambda_1,\lambda_l,\nu_1,\nu_n; \mu_1,\mu_m,\sigma_1,\sigma_s} = \alpha^{*}_{\lambda_1,\lambda_l,\nu_1,\nu_n; \mu_1,\mu_m,\sigma_1,\sigma_s} $, therefore, $\hat{P_A}\wedge \hat{P_B}$ is Hermitian.
\end{proof}

\begin{lem}
If $\hat{C}$ is a local operator on $A$ and $\hat{D}$ a local operator on $B$, then $\tr(\hat{C}\wedge D)=\tr(\hat{C}) \tr(D)$.
\label{lema:trprod}
\end{lem}
\begin{proof}
Consider the $\hat{C}$ and $\hat{D}$ as
\begin{gather}
\hat{C}=\sum_{\vec{\lambda}, \vec{\mu}} \eta_{\vec{\lambda}, \vec{\mu}} f_{\lambda_1}^{\dagger}\ldots f_{\lambda_l}^{\dagger} \proj{\Omega} f_{\mu_m}\ldots f_{\mu_1}, \quad \hat{D}=\sum_{\vec{\nu}, \vec{\sigma}} \beta_{\vec{\nu}, \vec{\sigma}} f_{\nu_1}^{\dagger}\ldots f_{\nu_n}^{\dagger} \proj{\Omega} f_{\sigma_s}\ldots f_{\sigma_1}.
\end{gather}
Then it can be seen that by the imposition of a specific order, as seen in the proof of Lemma \ref{lema:scprod}:
\begin{gather}
\hat{C}\wedge \hat{D}= \sum_{\vec{\lambda}, \vec{\mu},\vec{\nu}, \vec{\sigma}}  \eta_{\vec{\lambda}, \vec{\mu}} \beta_{\vec{\nu}, \vec{\sigma}}  f_{\lambda_1}^{\dagger}\ldots f_{\lambda_l}^{\dagger} f_{\nu_1}^{\dagger}\ldots f_{\nu_n}^{\dagger} \proj{\Omega} f_{\sigma_s}\ldots f_{\sigma_1}  f_{\mu_m}\ldots f_{\mu_1}, \\ \tr(\hat{C}\wedge \hat{D})= \sum_{\vec{\lambda}, \vec{\nu} } \eta_{\vec{\lambda}, \vec{\lambda}} \beta_{\vec{\nu}, \vec{\nu}}=\sum_{\vec{\lambda}} \eta_{\vec{\lambda},\vec{\lambda}} \cdot \sum_{\vec{\nu}} \beta_{\vec{\nu}, \vec{\nu}} =\tr(\hat{C}) \cdot \tr(\hat{D}). 
\end{gather}
\end{proof}

\begin{lem}
For a fermionic space $\mathcal{H}$ of $N$ modes labeled by $x_1,\ldots, x_{N}$, the projector $\proj{\Omega}$ in terms of the creation and annihilation operators is
\begin{equation}
\proj{\Omega}=f_{x_{N}}\ldots f_{x_1}f^{\dagger}_{x_1}\ldots f^{\dagger}_{x_{N}}
\end{equation}
\label{lema:0}
\end{lem}
\begin{proof}
 The proof can be easily followed from the structures of creation and annihilation operators. 
\end{proof}
This lemma enables us to say that any object, in the operator space of the Hilbert space $\mathcal{H}$, is a linear combination of products of creations and annihilation operators. Moreover, the following Corollary can be derived:

\begin{corollary}
If $\hat{O}_A$ is a fermionic local operator on the subspace spanned by the set of modes $A=\{a_1, \dots, a_M\}$, then $\hat{O}_A$ can be written as a sum of products of creation and annihilation operators of the modes in $A$ alone. 
\end{corollary}

\begin{prop*}
\textbf{\ref{Prop:basis1})} For an $N$-mode fermion system, the set $\mathcal{B}=\{\ket{\Omega},\ket{1},\ket{2},\ket{1}\wedge \ket{2},\ket{3},\ket{1}\wedge\ket{3}, \ket{2}\wedge \ket{3}, \ket{1}\wedge\ket{2}\wedge\ket{3},\ldots,  \ket{1}\wedge\ket{2}\wedge \ldots \wedge \ket{N} \}$ is an orthonormal basis of the Hilbert space $\mathcal{H}$, with dimension $2^N$.
\end{prop*}
\begin{proof}
We know that $\mathcal{H}=\bigoplus_{k=0}^{+\infty} \mathcal{H}_{k-p}$. For a $N$ mode fermionic system, we can see that $\mathcal{H}_{k-p}=\{0\}$ $\forall k>N$. This is due to the fact that these spaces are spanned by the elements $f_{i_1}^{\dagger} \dots f_{i_k}^{\dagger}\ket{\Omega}$, and since $k>N$ there must exist at least one repeated index $i_{l}=i_m$ for $1\leq l < m\leq k$. Thus, by anti-commuting this terms it is obtained $f_{i_1}^{\dagger} \dots f_{i_k}^{\dagger}\ket{\Omega}=(-1)^{l-m+1}f_{i_1}^{\dagger} \dots f_{i_l}^{\dagger}f_{i_m}^{\dagger} \dots  f_{i_k}^{\dagger}\ket{\Omega}= (-1)^{l-m+1}f_{i_1}^{\dagger} \dots f_{i_l}^{\dagger}f_{i_l}^{\dagger} \dots  f_{i_k}^{\dagger}\ket{\Omega}=0$. Where it is used that $f_i ^2=0=(f_i^{\dagger})^2$. So we have that $\mathcal{H}$ is the space spanned by all the elements $f_{i_1}^{\dagger}\ldots f_{i_k}^{\dagger}\ket{\Omega}$ where $k=0,\ldots,N$. Since the creation operators can be anti-commuted, WLOG it can be said that the space is spanned by the elements above such that $i_1<i_2<\ldots<i_k$. Because if are non-ordered, can be anti-commuted to an ordered case and strict inequality because if 2 are equal then the contribution is cancelled. By combinatorics is not difficult to see that the number of elements that span each $\mathcal{H}_{k-p}$ under this restrictions is $N\choose{k}$. So since the basis of $\mathcal{H}$ will be the reunion of al this generators, because it is a direct sum; we obtain that the dimension of $\mathcal{H}$ is $\sum_{k=0}^{N} {N\choose k} =2^N$ by the Newton binomial formula. So, if we find $2^N$ linearly independent states of $\mathcal{H}$ that are orthonormal we would have found a Hilbert basis of the space. To see that $\mathcal{B}$ is a basis, we start by observing that it can be constructed by starting with $\mathcal{B}_1=\{\ket{\Omega},\ket{1}\}$, then to each element we leave it invariant or apply the next creation operator, in this case $f^{\dagger}_2$, to obtain: $\mathcal{B}_2={\ket{\Omega},\ket{1},\ket{2},\ket{2}\wedge\ket{1}}$. In general the array $\mathcal{B}_{n+1}$ is obtained by concatenating $\mathcal{B}_{n}$ and $f^{\dagger}_{n+1}\mathcal{B}_{n}$. Then, once we have arrived to $\mathcal{B}_N$ we have obtained $2^N$ elements of the space $\mathcal{H}_{pure}$. By the construction, they are linearly independent because each one has a non repeating appearence of creation operators. In order that $\mathcal{B}_N$ is really the $\mathcal{B}$ of the proposition, it is required to reorder the components of the elements by reordering the creation operators, such operation only contributes up to a $-1$ sign that is irrelevant in the basis. So it is proven that $\mathcal{B}$ is a basis of $\mathcal{H}$. And to proof that is a Hilbert basis, it is only required to use the Lemma \ref{lema:scprod} and it follows directly from that result.     
\end{proof}

%\newpage 

\section{Definition of partial tracing} 
\label{sec:ptappendix}

In this Appendix \ref{sec:ptappendix} the Proposition \ref{prop:partialtrace} is proven by showing several Lemmas that lead to the complete proof. The proof does work for finite systems although it seems that the result can be generalized to the countable infinite case. A neat property of the partial tracing procedure is that we can easily implement it in matrix representations. In \cite{code} there is an implementation of the partial tracing procedure in matrix representation, where a Mathematica function is programmed to take the fermionic partial trace. In the following lines, we deduce the partial tracing procedure for fermionic systems under the SSR. We observe that the obtained procedure is the same as in \cite{Friis16} where the SSR condition was not part of the partial tracing definition.

\begin{lem}
\label{lema:tracezero}
If we have an hermitian operator $C\in \Gamma_S$ acting on $\mathcal{H}_S$ such that satisfies:
\begin{gather}
\tr(O \cdot C)=0 ,\ \ \forall O\in \mathcal{O}_S
\end{gather}
where $\mathcal{O}_S$ is the set of all hermitian operators of $\Gamma_S$. Then $C$ is the null operator, $C=0$. 
\end{lem}

\begin{proof}
Since the trace is invariant up to unitary transformations and being $C$ a hermitian matrix, there exist a unitary matrix $U$ such that $D=U\cdot C\cdot U^{\dagger}$ is a diagonal matrix with real values on it. Due to the fact that $C$ is block diagonal, the unitary $U$ can also be chosen to be block diagonal. With these impositions the conditions can be modified to:

\begin{gather}
\tr(U\cdot O \cdot U^{\dagger}\cdot U\cdot C\cdot U^{\dagger})=\tr(U\cdot O \cdot U^{\dagger}\cdot D)=0 \qquad \forall O \in \mathcal{O}
\end{gather}

Since $D$ is a diagonal matrix it follows that $[D]_{i,j}=\lambda_i \delta_{i j} \in \mathbb{R}$.\\

The matrix $M=U\cdot O \cdot U^{\dagger}$ is also block diagonal. It can be chosen a string of matrices $M_i$ for $i=1,\ldots 2^{N-1}$ such that each $M_i$ it takes the form of $[M_i]_{j,k}=\delta_{i k}\delta_{j i}$. Notice that each matrix $M_i$ is a real-valued diagonal matrix thus is hermitian and block diagonal. Therefore by setting $O=U^{\dagger}\cdot M_i\cdot U$, which is the unitary transformation (by $U^{\dagger}$) of a hermitian matrix, it will also be a hermitian matrix that is block-diagonal, and consequently $O\in \mathcal{O}_S$ as desired. So for each convenient $M_i$ exists a hermitian block-diagonal $O$ such that the condition has to hold. And therefore the string of impositions. 

\begin{equation}
\tr(M_i\cdot D)=0 \qquad \forall i\in\{1,\ldots, 2^{N-1}\}
\end{equation}
must be true. And therefore since 

\begin{equation}
[M_i \cdot D]_{j ,k} =\delta_{j i}\delta_{i k}\lambda_i
\end{equation}
 the conditions become:
 
 \begin{equation}
 \tr(M_i\cdot D)=\lambda_i=0  \qquad \forall i\in\{1,\ldots, 2^{N-1}\}
 \end{equation}
and therefore it must be that $D=0$ and since $C=U^{\dagger}\cdot D\cdot U$, $C=0$.
\end{proof}

\begin{lem} 
\label{lema:nth}
If the partial trace of the $N^{th}$ mode is well defined and unique, then the definition of the partial trace on an arbitrary mode is well defined and unique.
\end{lem}
\begin{proof}
Given an operator $A$ of the fermionic space for $N$ modes, labelled as $\{\ket{1},\ldots,\ket{N}\}$. Then $A$ can be decomposed as: $A=\sum_{\vec{r} \vec{s}} \alpha_{\vec{r} \vec{s}} \ket{1}^{r_1}\wedge \ldots \wedge \ket{N}^{r_N}\bra{1}^{s_1}\wedge \bra{N}^{s_N} $. If the mode $j$ it is wanted to be traced out, since the operator $A$ can be written as  $A=\sum_{\vec{r} \vec{s}} (-1)^{\delta}\alpha_{\vec{r} \vec{s}} \ket{1}^{r_1}\wedge \ldots \wedge \ket{N}^{r_N}\wedge\ket{i}^{r_i} \bra{1}^{s_1}\wedge \bra{N}^{s_N}\wedge \bra{i}^{s_i}$ where $\delta$ is an integer that depends on $s_i,r_i,\vec{s},\vec{r}$. Since the partial trace of the last mode it is well defined and unique and now the $i$ mode is the last, the partial trace over the $i$ mode is well defined and unique. Since $i$ is arbitrary, the partial trace on any mode is well defined and unique.      
\end{proof}

\begin{lem}
\label{lema:ordering}
If the partial tracing of one mode is well defined unique, then partial tracing over an ordered set of modes $k_1,\ldots,k_l$ is well defined and unique so that:
\begin{gather}
\tr_{k_1,\ldots,k_l}=\tr_{k_1}\circ \ldots\circ \tr_{k_l}
\end{gather}
\end{lem}

\begin{proof}
Let us define that $\Gamma$ is the global operator space with $N>l$ modes, and denote by $\mathcal{H}_C$ as the Hilbert fermionic space where we remove the $k_1,\ldots,k_l$ modes from the initial set, $\mathcal{H}_{k_1,\ldots,k_m}$ the Hilbert space of all $N$ modes. 
To proof the equality of these 2 operations lets proof that they act the same way to an arbitrary state $\rho$. So if it is seen that $\tr_{k_1,\ldots,k_l}(\rho)=\tr_{k_1}\circ \ldots\circ \tr_{k_l}(\rho)$  it is done. By choosing the ordering of the modes as putting $k_1,\ldots,k_l$ the last ones, the operators on $\Gamma_S$ that are local on $\mathcal{H}_C$ can be written as $O_C\wedge \mathbb{I}\wedge \ldots \wedge\mathbb{I}$. 
From the definition of partial trace it is known that:
\begin{gather}
\tr( (O_C\wedge \mathbb{I}\wedge \ldots \wedge\mathbb{I}) \rho)=\tr(O_C \tr_{k_1,\ldots,k_l}(\rho)) \qquad \forall O_C\in \mathcal{O}^C_S
\end{gather}
where $\mathcal{O}^C_S$ is the set of hermitian local operators in $\Gamma^C_S$.
Now, in the other case, since $\tr_{k_m}( \ldots(\tr_{k_l}(\rho))\ldots)\in \mathcal{H}_{k_1\ldots,k_{m-1}}$ is clear that since $\mathcal{O}_{C,k_1,\ldots, k_n}\subset \mathcal{O}_{C,k_1,\ldots, k_m} $ if $n\leq m$, for the definition of the partial trace operation: 
\begin{align}
\tr(O_C \tr_{k_1}( \ldots(\tr_{k_l}(\rho))\ldots))&=\tr((O_C \wedge \mathbb{I}) \tr_{k_2}( \ldots(\tr_{k_l}(\rho))\ldots)) \nonumber \\ 
&=\tr((O_C \wedge \mathbb{I} \wedge \mathbb{I}) \tr_{k_3}( \ldots(\tr_{k_l}(\rho))\ldots))\nonumber \\ 
& \ \ \vdots \nonumber \\ 
&= \tr((O_C \wedge \mathbb{I}\wedge \ldots \wedge \mathbb{I}) \tr_{k_l}(\rho)) \nonumber \\ 
&=\tr((O_C \wedge \mathbb{I}\wedge \ldots \wedge \mathbb{I} ) \rho) \quad \forall O_C \in \mathcal{O}^C_S 
\end{align}
And therefore it is found that
\begin{equation}
\label{eq:ordering}
\tr\left(O_C \left[\tr_{k_1}( \ldots(\tr_{k_l}(\rho))\ldots)-\tr_{k_1,\ldots,k_l}(\rho)\right]\right)=0 \qquad \forall O_C \in \mathcal{O}^C_S 
\end{equation}
And for what it has been shown in the Lemma \ref{lema:tracezero}, And since $\rho$ has been arbitrary, this condition is sufficient to say that $\tr_{k_1}\circ  \ldots \circ \tr_{k_l} = \tr_{k_1,\ldots,k_l}$.
\end{proof}

Is important to remark that the uniqueness of the partial tracing procedure does not rely on the concrete ordering seen in the Lemma \ref{lema:ordering} since by the definition of partial tracing and its uniqueness is straightforward to prove that $\tr_{A} \circ \tr_{B}=\tr_{B} \circ \tr_{A}$. Nevertheless, for procedural purposes is usually easier to trace out the 'largest' mode first.

\begin{prop*}\textbf{\ref{prop:partialtrace}}) For a density operator $\rho \in \mathcal{R}^N_S$, of an $N$-mode fermionic system, the partial tracing over the set of modes $M=\{m_1,\ldots,m_{|M|}\}\subset\{1,\ldots,N\}$ must result in a reduced density operator $\sigma \in \mathcal{R}^{N\setminus M}_S$, given by
\begin{align}
\sigma=\tr_{M}(\rho)=\tr_{m_1}\circ \tr_{m_2} \circ \ldots \circ \tr_{m_{|M|}}(\rho),
\end{align}
 There is a unique partial tracing procedure that satisfies the physically imposed consistency conditions. The operation of partial tracing one mode $m_i$, of an element of $\rho$, is given then by:
\begin{align}
\tr_{m_i} & \left(\left(f_1^{\dagger}\right)^{s_1} \ldots \left(f_{m_i}^{\dagger}\right)^{s_{m_i}}\ldots \left(f_N^{\dagger}\right)^{s_N}\proj{\Omega}f_N^{r_N}\ldots f_{m_i}^{r_{m_i}}\ldots f_1^{r_1} \right) = \delta_{s_{m_i} r_{m_i}} (-1)^{k}\left(f_1^{\dagger}\right)^{s_1} \ldots \left(f_N^{\dagger}\right)^{s_N}\proj{\Omega}f_N^{r_N}\ldots f_1^{r_1}, 
\end{align}
with $k=s_{m_i} \sum_{j=m_i}^{N-1} s_{j+1} +r_{m_i} \sum_{k=m_i}^{N-1} r_{k+1}$ and $s_i,r_j \in \{0,1 \}$.
\end{prop*}

\begin{proof}
 By the Lemmas \ref{lema:nth} \& \ref{lema:ordering} it is enough to see that there exists a well defined and unique way to trace out the $N$th mode and that it corresponds to the tracing procedure proposed in the statement of the proposition. \\

First we observe that to say that: 

\begin{gather}
\tr_{m_i}\left(\left(f_1^{\dagger}\right)^{s_1} \ldots \left(f_{m_i}^{\dagger}\right)^{s_{m_i}}\ldots \left(f_N^{\dagger}\right)^{s_N}\proj{\Omega}f_N^{r_N}\ldots f_{m_i}^{r_{m_i}}\ldots f_1^{r_1} \right) = \delta_{s_{m_i} r_{m_i}} (-1)^{k}\left(f_1^{\dagger}\right)^{s_1} \ldots \left(f_N^{\dagger}\right)^{s_N}\proj{\Omega}f_N^{r_N}\ldots f_1^{r_1} 
\end{gather}
where $k=s_{m_i} \sum_{j=m_i}^{N-1} s_{j+1} +r_{m_i} \sum_{k=m_i}^{N-1} r_{k+1}$ and $s_i,r_j$ take the value $0$ or $1$. Is equivalent to say that you put the $m_i$ mode as the $N$th mode and then generate the basis $\mathcal{B}$, where the matrix representation of $A$ will be $\left.A\right|_{\mathcal{B}}=\begin{pmatrix}
a && b \\ c && d
\end{pmatrix}$  then $\left.\tr_{m_i} A\right|_{\mathcal{B}}=a+d$.\\

This form of expressing the partial trace is useful to prove that indeed is a well defined partial tracing procedure.\\ 

First, since the partial tracing is a linear operation and due to the correspondence of the matrix representation of the operators with the operator space; if the partial trace can be defined and seen to be unique in terms of a matrix representation for a concrete basis, then by the correspondence of matrix representations with the space itself, the operation will be well defined and unique in the linear space. \\

As a first step, we check that the found procedure satisfies the properties of a partial trace. Suppose a hermitian local operator $O_A\in \Gamma_S^A$ on the first $N-1$ modes of the system. We can easily see that $\tilde{O}_A\otimes \mathbb{I}_2$ is the matrix representation of $O_A$ on the basis $\mathcal{B}$. Where $ \tilde{O}_A$ is the matrix representation of the operator $O_A$ in the basis $\mathcal{B}$ restricted to the space of the $N-1$ first modes. Now, it is wanted to check that defining that $\tilde{\rho}_A=a+c$ (in the decomposition seen above) it follows the conditions of partial tracing:

\begin{equation}
\tr(O_A\cdot\rho)=\tr(O_A\cdot \rho_A) \quad \forall O_A\in \mathcal{O}^A_S
\end{equation}

(where $\mathcal{O}^A_S$ is the set of all local hermitian operators in $\Gamma_S^A$. So let's check it:

\begin{gather}
\tr(O_A \cdot \rho)=\tr((\tilde{O}_A \otimes \mathbb{I}_2 )\left.\rho\right|_{\mathcal{B}})=\tr\left(\left(\begin{matrix}
\tilde{O}_A && 0 \\ 0 && \tilde{O}_A
\end{matrix}\right) \left(\begin{matrix}
a && b \\ b^* && c
\end{matrix}\right)\right)=\tr\left(\begin{matrix}
\tilde{O}_A \cdot a && \tilde{O}_A \cdot b\\ \tilde{O}_A\cdot b^* && \tilde{O}_A \cdot c
\end{matrix}\right) = \tr(\tilde{O}_A \cdot a )+\tr(\tilde{O}_A \cdot c)=\nonumber\\=\tr(\tilde{O}_A \cdot a +\tilde{O}_A \cdot c )= \tr(\tilde{O}_A \cdot (a + c) )=\tr(\tilde{O}_A \cdot \tilde{\rho}_A )=
\end{gather}
and since $O_A$ is arbitrary, it is proven by all $O_A\in \mathcal{O}^A_S$. \\

So it just has been seen that a partial trace operation can be defined. Now the second step is to see that this definition is unique. \\

Let's assume that there exists another consistent definition of a partial trace, that will give rise to $\rho'_A$ as a reduce state. Lets proof that $\tilde{\rho'}_A-a-c=0$. Since for both notions the consistency conditions are true, the following relations hold:

\begin{gather}
\tr(\tilde{O}_A \cdot (a+c))=\tr(\tilde{O}_A\otimes \mathbb{I}_2\cdot \tilde{\rho})=\tr(\tilde{O}_A \cdot \tilde{\rho'}_A) \quad \forall O_A \in \mathcal{O}^A_S \qquad \tr(\tilde{O}_A \cdot( \tilde{\rho'}_A-a-c))=0 \qquad \forall O_A \in \mathcal{O}_S^A 
\end{gather}

Since $\tilde{\rho'}_A$ and $a+c$ are matrix representations of reduced states and therefore states, both are hermitian and block-diagonal. Therefore the matrix $C=\tilde{\rho'}_A-a-c$ is hermitian and block-diagonal. The condition obtained becomes the same as in Lemma \ref{lema:tracezero}, and thus $C=0$ and so $\tilde{\rho'}_A=a+c$. Thus proving the uniqueness of the procedure, and hence proving the proposition.

\end{proof}

As we comment in the main text, one could think that the introduction of the SSR attempts against the uniqueness of the partial tracing procedure. However, we have seen that the partial tracing procedure does not lose its uniqueness nor the form used in \cite{Friis16} for general fermionic state contributions. Imposing the SSR to states is compensated by its imposition to the observables.

\section{Incompatibility between Jordan-Wigner transformation and partial tracing operation}
\label{sec:incopatibility}

Why such development is important since we could use the Jordan-Wigner transformation to transform the fermionic states of $N$ modes to $N$ qubit states? As shown in \cite{Friis16}, under the parity SSR, the Jordan-Wigner transformation does not correctly describe partial tracing operation. In other words, by choosing any concrete Jordan-Wigner transformation $JW: \mathcal{H}_S^N\to \mathbb{C}^N$, we can find SSR fermionic states $\rho_{AB}$ such that $\tr_B(JW \cdot \rho_{AB} \cdot JW^{\dagger}) \neq JW\cdot \tr_B(\rho_{AB}) \cdot JW^{\dagger}$. Thus, when proceeding to work with $\rho_A$ it is unknown which representation to take in this scenario. However, even with this inconsistency, some can say that it is just a defect of the representation and that this inconsistency is minor and does not play a role in the calculation of relevant physical quantities. In the following lines, we illustrate how this is not the case for the canonical Jordan-Wigner transformation; there are SSR states that given these two procedures, the Von Neumann entropy of the associated $\rho_A$ is different. 

The Von Neumann entropy is a measure of information in the usual sense, for any state in $\mathcal{R}_{S}$, thanks to the diagonalization Proposition  \ref{prop:observables}. Any state $\rho\in \mathcal{R}_{S}$ can be diagonalized,  $\rho=\sum_i p_i \proj{\psi_i}$ with the eigenvalues being the probabilities $\{p_i\}$ and the eigenvectors being pure SSR states. Thus the Von Neumann entropy can be computed accordingly as 
\begin{align}
	S(\rho)=-\tr \rho \log \rho =-\sum_i p_i \log p_i,
\end{align}
where we have used base $2$ with $\log 2=1$.

Given the space of $4$ fermionic modes with canonical base $\mathcal{B}$ where and the canonical Jordan-Wigner transformation given by the assignment  $f_i\to -(\bigotimes_{k=1}^{i-1} \hat{\sigma}_z )\otimes \ketbra{0}{1} \otimes (\bigotimes_{j=i+1}^4 \mathbb{I})$. this assignment is such that it provides the assignments $\ket{\Omega} \to \ket{0000}, \ket{1} \to -\ket{1000}, \ket{2}\to -\ket{0100}, \ket{1}\wedge\ket{2} \to \ket{1100}, \dots $. Notice that the odd states pick a global $-1$ phase. We notice that for a SSR state, in the density matrix formalism these $-1$ phases disappear. Thus, by choosing the canonical basis of the qubit space ${\mathbb{C}^2}^{\otimes 4}$ the canonical Jordan-Wigner transformation corresponds to the identity mapping of the representation of the density matrices in the corresponding canonical basis.   

With that mapping, and choosing the mode bipartition $14|23$ we can see that the fermionic SSR density operator $\rho_{1234}$ that is given by
\begin{gather}
	\rho_{1234}=\frac{1}{16} \left[\mathbb{I} + \ket{\Omega}(\bra{1}\wedge\bra{4}) + \ket{1}\bra{4} + \ket{2}(\bra{1}\wedge\bra{2}\wedge\bra{4}) + \right. \nonumber \\ + (\ket{1}\wedge \ket{2})(\bra{2}\wedge\bra{4}) + \ket{3}(\bra{1}\wedge\bra{3} \wedge \bra{4}) +\nonumber \\+ (\ket{1}\wedge \ket{3})(\bra{3}\wedge\bra{4}) + (\ket{2}\wedge \ket{3})(\bra{1}\wedge\bra{2}\wedge\bra{3}\wedge\bra{4}) +\nonumber \\ \left. +(\ket{1}\wedge \ket{2}\wedge\ket{3})(\bra{2}\wedge\bra{3}\wedge\bra{4}) + h.c.  \right]
\end{gather}
yields the following inconsistent results. The representation on the canonical computation basis, eigenvalues and Von Neumann information for the operator obtained by taking the fermionic partial trace to obtain $\rho_{12}$ and then applying the Jordan-Wigner transformation are the following
\begin{gather}
	\left.JW\cdot \rho_{12}\cdot JW \right|_{\text{comp},\text{comp}} = \frac{1}{4} \begin{pmatrix}
		1 & 0& 0& 0\\ 0 &1 &0 &0\\ 0 &0 &1 &0 \\ 0& 0 &0& 1
	\end{pmatrix} \nonumber \\ \lambda_i\in \left\{\frac{1}{4},\frac{1}{4},\frac{1}{4},\frac{1}{4} \right\} \quad S(\rho_{12})= 2
\end{gather}

while on the other hand, if we first take the Jordan-Wigner transformation and then apply the qubit partial tracing procedure the results become:
\begin{gather}
	\left.\tr_{23}(JW\cdot\rho_{1234}\cdot JW)\right|_{\text{comp},\text{comp}} = \frac{1}{4} \begin{pmatrix}
		1 & 0& 0& 1\\ 0 &1 &1 &0\\ 0 &1 &1 &0 \\ 1& 0 &0& 1
	\end{pmatrix} \nonumber \\ \lambda_i\in \left\{\frac{1}{2},\frac{1}{2},0,0 \right\} \quad S(\rho_{12})= 1
\end{gather}

Therefore, showing the complete inconsistency of this use of the Jordan-Wigner function. Due to the excellent properties of the canonical Jordan-Wigner function on states under the SSR, the correct result using purely fermionic based representations is the same as the one yielded in the first procedure of the two.

%\newpage

\section{Proofs of the structures that arise by imposing the parity super-selection rule}
\label{sec:proofs}
As shown in the main text, we invoke some restrictions on the fermionic space - the parity super-selection rule (SSR). In the following lines, we prove the found properties analogous to the tensor product space case described in the main text.

\begin{lem}
\label{lema:ssrmixed}
The matrix representation of $\rho$, in the basis $\mathcal{B}'$, takes the form
\begin{align}
\rho= \rho_e \oplus \rho_o =\begin{pmatrix}
\rho_e && 0 \\ 0 && \rho_o
\end{pmatrix},
\end{align}
where $\rho_e$ ($\rho_o$) are $2^{N-1}\times 2^{N-1}$ complex valued matrices. 
\end{lem}

\begin{proof}
Since $\rho$ is by definition an ensemble of pure fermionic SSR states then its expression in equation \ref{eq:rho} can be broken as
\begin{gather}
\rho=\sum_i p_i \sum_{\vec{s^i},\vec{r^i}; R^i,S^i\in even} \alpha_{\vec{s^i}} \alpha^*_{\vec{r^i}} \left(f^{\dag}_1\right)^{s^i_1} \ldots (f^{\dag}_N)^{s^i_N} \proj{\Omega} f_N^{r^i_N} \ldots f_1^{r^i_1} +\sum_i p_i\sum_{\vec{s^i},\vec{r^i}; R^i,S^i\in odd} \alpha_{\vec{s^i}} \alpha^*_{\vec{r^i}} \left(f^{\dagger}_1\right)^{s^i_1} \ldots \left(f^{\dagger}_N\right)^{s^i_N} \proj{\Omega} f_N^{r^i_N} \ldots f_1^{r^i_1} \nonumber \\
=\rho_E \oplus \rho_O
\end{gather}
and since the basis $\mathcal{B}'$ is a rearrangement of the basis $\mathcal{B}$ (that is constituted by the states $\left(f_1^\dagger\right)^{s_1}\ldots \left(f_N^\dagger\right)^{s_N}\ket{\Omega}$), where the $2^{N-1}$ even and $2^{N-1}$ odd components are splitted; we see that the $\rho_E$ and $\rho_O$ components will clearly have a matrix representation in the basis $\mathcal{B}'$ as 
\begin{align}
\left.\rho_E\right|_{\mathcal{B}'}=\begin{pmatrix}
E_{\rho} && 0 \\ 0 && 0
\end{pmatrix} \qquad \qquad\left.\rho_O\right|_{\mathcal{B}'}=\begin{pmatrix}
0 && 0 \\ 0 && O_{\rho}
\end{pmatrix}
\end{align}
The sum of matrix representations is the matrix representation of the sum, therefore the result follows. 
\end{proof}

\begin{prop}
\label{prop:ssrmixed}
$\rho$ is a SSR density operator iff $\rho\in \Gamma_S$ is a positively semi-defined Hermitian operator with trace one.
\end{prop}

\begin{proof}
$\Rightarrow$: since $\rho=\sum_i p_i \proj{\psi_i}$ with $\ket{\psi_i}$ being SSR states and $0\leq p_i \leq 1$ $\sum_i p_i =1$, we have by Lemma \ref{lema:ssrmixed} directly that $\rho \in \Gamma_S$. Of this form is also directly deduced that $\rho$ is Hermitian and that $\tr(\rho)=1$. To see positivity we only need to check that $\bra{\varphi} \rho \ket{\varphi}=\sum_i p_i |\scp{\varphi}{\psi_i}|^2 \geq 0$ for all $\ket{\varphi}$. \\

$\Leftarrow$: We use the matrix representation of a positively semi-definite Hermitian operator of $\Gamma_S$ with trace 1 in the basis $\mathcal{B}'$. Such operator $\rho$ will have the following matrix representation:
\begin{gather}
    \left. \rho \right|_{\mathcal{B}'}= \begin{pmatrix}
    \rho_e & 0 \\ 0 & \rho_o
    \end{pmatrix}
\end{gather}
The fact that $\rho$ is Hermitian implies that $\rho_e$ and $\rho_o$ have to be Hermitian matrices. This implies that they can be diagonalised. This means that there exist different changes of basis within the subspaces spanned by $\mathcal{B}_e$ and $\mathcal{B}_o$ such that without changing the block-parity structure exist basis $\tilde{\mathcal{B}}_e$ and $\tilde{\mathcal{B}}_o$ such that in the total basis $\tilde{\mathcal{B}}'=\tilde{\mathcal{B}}_e \cup \tilde{\mathcal{B}}_e$ the matrix representation of $\rho$ would be
\begin{gather}
    \left. \rho \right|_{\tilde{\mathcal{B}}'} = \begin{pmatrix}
    \lambda_1 &  & & 0 & & \\
     & \ddots & & & \ddots & \\
      & & \lambda_{2^{N-1}} & & & 0 \\
      0 & & & \mu_1 & & \\
       & \ddots & & & \ddots & \\
       & & 0 & & & \mu_{2^{N-1}}
    \end{pmatrix}
\end{gather}
Now, since $\rho$ is positively semi-definite this implies that necessarily $\mu_i,\lambda_j \geq 0$ for all $i,j$. Thus, this means that $\rho$ can be written as $\rho=\sum_i \lambda_i \proj{\psi_i} + \sum_j \mu_j \proj{\varphi_j}$ where $\ket{\psi_i}$ are even states that conform an orthonormal basis of the even subspace, and equally for $\ket{\varphi_j}$ in the odd case. This means that the union of the two conforms an orthonormal basis  $\{\ket{\tilde{\psi}_i}\}_{i=1}^{2^N}$ where the first $2^{N-1}$ elements are evens states and the last are odd states. Thus by choosing $p_i=\lambda_i$ and $p_{2^{N-1}+i}=\mu_i$ for $i=1,\dots, 2^{N-1}$, we obtain that $\rho=\sum_i p_i \proj{\tilde{\psi}_i}$. From the definition of $p_i$ and the positive semi-definitness of $\rho$ it follows that $p_i\geq 0$. Since $\tr(\rho)=1$ it follows that $\sum_i (\lambda_i + \mu_i)=1$ thus giving $\sum_i p_i=1$. Thus fulfilling the requirements of being a density operator.    
\end{proof}

We remind the reader that we name the set of density operators that follow the SSR as $\mathcal{R}_S$. In the next lines, we prove a general statement of the block form that the fermionic linear operators have to preserve the SSR structure. We denote them by SSR linear operators.  

\begin{theorem*}
\textbf{\ref{thm:blockform}. Linear operators})
If $\hat{O}$ is a linear operator such that $\hat{O}: \mathcal{H}^N \rightarrow \mathcal{H}^M $ such that for all $\hat{A}\in \Gamma_S^N$ and $\hat{B}\in \Gamma_S^M$,  $\hat{O} \hat{A}  \hat{O}^\dag \in \Gamma_S^M$ and $\hat{O}^\dagger \hat{B} \hat{O} \in \Gamma_S^N$ then the matrix representation of the operator $\hat{O}$ on the basis ${\mathcal{B}'}^N, {\mathcal{B}'}^M $ has to take one of these two diagonal and anti-diagonal forms:

\begin{gather}
\hat{O}|_{{\mathcal{B}'}^M,{\mathcal{B}'}^N}=\left(\begin{array}{c|c}
O_{++}  & \hat{0} \\ \hline \hat{0}  & O_{--}  
\end{array}\right)  \medspace \medspace \text{or} \medspace \medspace \hat{O}|_{{\mathcal{B}'}^M,{\mathcal{B}'}^N}=\left(\begin{array}{c|c}
\hat{0} & O_{+-}  \\ \hline  O_{-+}  &  \hat{0}
\end{array}\right)
\end{gather}
where $O_{++},O_{--},O_{-+} ,O_{+-}  \in M_{2^{M-1}\times 2^{N-1}}(\mathbb{C})$  and $\hat{0}$ is the zero element of $M_{2^{M-1}\times 2^{N-1}}(\mathbb{C})$. Such operators are linear operators that preserve the SSR form of the operators in $\Gamma_S$
\end{theorem*}

\begin{proof}
To start, since $\hat{O}$ is a linear operator between the two spaces, it can be represented in the most general way in the basis $\mathcal{B}'^{N}$ and $\mathcal{B}'^{M}$ as:
\begin{gather}
\left.\hat{O}\right|_{\mathcal{B}'^{M}, \mathcal{B}'^{N}}=\begin{pmatrix}
O_{++} && O_{+-} \\ O_{-+} && O_{--}
\end{pmatrix}
\end{gather}
But since it is required that preserves the separation among SSR and non-SSR contributions we have that, by choosing the decomposition of $\hat{A}\in \Gamma_S^N$ and $\hat{B}\in \Gamma_S^M$ in the basis $\mathcal{B}'^N$  and $\mathcal{B}'^M$ respectively, that we know are given by:
\begin{gather}
   \left.\hat{A}\right|_{\mathcal{B}'^{N}, \mathcal{B}'^{N}}=\begin{pmatrix}
A_{++} && \hat{0} \\ \hat{0} && A_{--}
\end{pmatrix} \qquad \qquad \left.\hat{B}\right|_{\mathcal{B}'^{M}, \mathcal{B}'^{M}}=\begin{pmatrix}
B_{++} && \hat{0} \\ \hat{0} && B_{--}
\end{pmatrix}
\end{gather}
Now, choosing the two spacial cases for each condition, by setting $A_{++}=\hat{0}$, $B_{++}=\hat{0}$ and in the other case setting $A_{--}=\hat{0}$, $B_{--}=\hat{0}$; we obtain that
\begin{gather}
    \left.\hat{O}\hat{A}\hat{O}^\dagger\right|_{\mathcal{B}'^{M}, \mathcal{B}'^{M}}=\begin{pmatrix}
O_{++} && O_{+-} \\ O_{-+} && O_{--} \end{pmatrix} \begin{pmatrix}
\hat{0} && \hat{0} \\ \hat{0} && A_{--} \end{pmatrix} \begin{pmatrix}
O^\dagger_{++} && O^\dagger_{-+} \\ O_{+-}^\dagger && O^\dagger_{--} \end{pmatrix}=  \begin{pmatrix}
O_{+-} A_{--} O^\dagger_{+-} && O_{+-} A_{--} O^\dagger_{--} \\ O_{--} A_{--} O^\dagger_{+-} && O_{--} A_{--} O^\dagger_{--} \end{pmatrix} \\
\left.\hat{O}\hat{A}\hat{O}^\dagger\right|_{\mathcal{B}'^{M}, \mathcal{B}'^{M}}=\begin{pmatrix}
O_{++} && O_{+-} \\ O_{-+} && O_{--} \end{pmatrix} \begin{pmatrix}
A_{++} && \hat{0} \\ \hat{0} && \hat{0} \end{pmatrix} \begin{pmatrix}
O^\dagger_{++} && O^\dagger_{-+} \\ O_{+-}^\dagger && O^\dagger_{--} \end{pmatrix}=  \begin{pmatrix}
O_{++} A_{++} O^\dagger_{++} && O_{++} A_{++} O^\dagger_{-+} \\ O_{-+} A_{++} O^\dagger_{++} && O_{-+} A_{++} O^\dagger_{-+} \end{pmatrix} \\
 \left.\hat{O}^\dagger \hat{B}\hat{O}\right|_{\mathcal{B}'^{N}, \mathcal{B}'^{N}}= \begin{pmatrix}
O^\dagger_{++} && O^\dagger_{-+} \\ O_{+-}^\dagger && O^\dagger_{--} \end{pmatrix} \begin{pmatrix}
\hat{0} && \hat{0} \\ \hat{0} && B_{--} \end{pmatrix} \begin{pmatrix}
O_{++} && O_{+-} \\ O_{-+} && O_{--} \end{pmatrix}=  \begin{pmatrix}
O_{-+}^\dagger B_{--} O_{-+} && O_{-+}^\dagger B_{--} O_{--} \\ O^\dagger_{--} B_{--} O_{-+} && O^\dagger_{--} B_{--} O_{--} \end{pmatrix} \\ 
\left.\hat{O}^\dagger\hat{B}\hat{O}\right|_{\mathcal{B}'^{N}, \mathcal{B}'^{N}}= \begin{pmatrix}
O^\dagger_{++} && O^\dagger_{-+} \\ O_{+-}^\dagger && O^\dagger_{--} \end{pmatrix} \begin{pmatrix}
B_{++} && \hat{0} \\ \hat{0} && \hat{0} \end{pmatrix} \begin{pmatrix}
O_{++} && O_{+-} \\ O_{-+} && O_{--} \end{pmatrix}=  \begin{pmatrix}
O^\dagger_{++} B_{++} O_{++} && O^\dagger_{++} B_{++} O_{+-} \\ O^\dagger_{+-} B_{++} O_{++} && O^\dagger_{+-} B_{++} O_{+-} \end{pmatrix}
\end{gather}
In order to preserve the SSR form for these two cases for each condition, we find 4 constraints for each condition:
\begin{gather}
    O_{-+} A_{++} O^\dagger_{++}=O_{++} A_{++} O^\dagger_{-+}=O_{--} A_{--} O^\dagger_{+-}= O_{+-} A_{--} O^\dagger_{--}=\hat{0} \qquad \forall A_{--}, A_{++} \in M_{2^{N-1}}(\mathbb{C}) \\  O^\dagger_{+-} B_{++} O_{++}= O^\dagger_{++} B_{++} O_{+-}= O^\dagger_{--} B_{--} O_{-+}= O^\dagger_{-+} B_{--} O_{--}=\hat{0} \qquad \forall B_{--}, B_{++} \in M_{2^{M-1}}(\mathbb{C})
\end{gather}
These restrictions can be reduced to two for each condition due to the $\dagger$ property.
\begin{gather}
    O_{-+} A O^\dagger_{++} = O_{+-} A O^\dagger_{--}=\hat{0} \qquad \forall A \in M_{2^{N-1}}(\mathbb{C}) \\
    O^\dagger_{+-} B O_{++} = O^\dagger_{-+} B O_{--}=\hat{0} \qquad \forall B \in M_{2^{M-1}}(\mathbb{C})
\end{gather}
In order to proceed is necessary to check that if exist $D\in M_{m,n}(\mathbb{C})$, $F\in M_{n,m}(\mathbb{C})$ such that $D\cdot E\cdot F=0$ for all $E\in M_{n}(\mathbb{C})$ then either $D=0$ or $F=0$. To prove it, lets assume there exist $D,F$ such that $\exists i_1,j_2\in \{1,\dots,m\}$, $\exists j_1,i_2\in \{1,\dots,n\}$ such that $D\cdot E\cdot F=0$ for all $E\in M_{n}(\mathbb{C})$ with $D_{i_1,j_1}\neq 0$ and $F_{i_2,j_2}\neq 0$. Now since it can be chosen the matrix $E$ given by $E_{i,j}=\delta_i^{j_1} \delta_j^{i_2}$ we find that ${(D\cdot E\cdot F)}_{i_1,j_2}=D_{i_1,j_1} F_{i_2,j_2} \neq 0$ giving us a contradiction. Thus the assumption must be false. \\

This result gives us that the overall four restrictions that we had transform to the following four statements. 1: Either $O_{-+}=0$ or $O_{++}=0$. 2: Either $O_{+-}=0$ or $O_{--}=0$. 3: Either $O_{+-}=0$ or $O_{++}=0$. 4: Either $O_{-+}=0$ or $O_{--}=0$. It follows that there are only two non-trivial configurations: $O_{+-}=0=O_{-+}$ or $O_{++}=0=O_{--}$. Such configurations exactly correspond to the diagonal and anti-diagonal block form showed in the theorem, just as desired. 
\end{proof}

We observe that all the elements of $\Gamma_S$ are SSR linear operators but that there can exist anti-diagonal operators that are not in $\Gamma_S$ that preserve the SSR structure. With this observation and classification, we move to prove the following result that makes the use of the partial tracing procedure under the SSR easier.

\begin{prop}
\label{prop:pureptssr} Assume that $\ket{a},\ket{b},\ket{c},\ket{d}\in \mathcal{H}_{S}$ for $N$ modes. Being $M\subset \{1,\dots, N\}$ where $\ket{a}$ and $\ket{b}$ only have contributions of the modes on $M$, and $\ket{c}$ and $\ket{d}$ only from $M^c$. Then we have that:
\begin{gather}
\tr_{M}\left(\ketbra{a}{b}\wedge \ketbra{c}{d}\right)=\scp{b}{a}\ketbra{c}{d} \qquad \qquad \tr_{M^c}\left(\ketbra{a}{b}\wedge \ketbra{c}{d}\right)=\scp{d}{c}\ketbra{a}{b} 
\end{gather}
\end{prop}

\begin{proof}
First let's develop  $\ketbra{a}{b}\wedge \ketbra{c}{d}$ in terms of the mode operators: 
\begin{gather}
\ket{a}=\sum_{\vec{a}} \alpha_{\vec{a}} \cdot  f^{\dagger}_{a_1}\dots f^{\dagger}_{a_{n}}\ket{\Omega}\quad \ket{b}= \sum_{\vec{b}} \beta_{\vec{b}} \cdot f^{\dagger}_{b_1}\dots f^{\dagger}_{b_{m}}\ket{\Omega} \quad \ket{c}=\sum_{\vec{c}} \gamma_{\vec{c}} \cdot f^{\dagger}_{c_1}\dots f^{\dagger}_{c_{r}}\ket{\Omega} \quad \ket{d}=\sum_{\vec{d}} \delta_{\vec{d}} \cdot f^{\dagger}_{d_1}\dots f^{\dagger}_{d_{s}}\ket{\Omega} \\  \ketbra{a}{b}\wedge \ketbra{c}{d}=\sum_{\vec{a},\vec{b},\vec{c},\vec{d}} \alpha_{\vec{a}} \beta^*_{\vec{b}} \gamma_{\vec{c}} \delta^*_{\vec{d}}  f^{\dagger}_{a_1}\dots f^{\dagger}_{a_{n}} f^{\dagger}_{c_1}\dots f^{\dagger}_{c_{r}}\proj{\Omega} f_{d_{s}}\dots f_{d_1} f_{b_{m}}\dots f_{b_1} \\ \tr_{M^c}\left(\ketbra{a}{b}\wedge \ketbra{c}{d}\right)=\sum_{\vec{a},\vec{b},\vec{c},\vec{d}} \alpha_{\vec{a}} \beta^*_{\vec{b}} \gamma_{\vec{c}} \delta^*_{\vec{d}}  \left(f^{\dagger}_{a_1}\dots f^{\dagger}_{a_{n}} \proj{\Omega}  f_{b_{m}}\dots f_{b_1} \right)\left(\bra{\Omega}f_{d_{s}}\dots f_{d_1}f^{\dagger}_{c_1}\dots f^{\dagger}_{c_{r}}\ket{\Omega}\right)=\nonumber \\=   \left(\sum_{\vec{a}} \alpha_{\vec{a}} f^{\dagger}_{a_1}\dots f^{\dagger}_{a_{n}} \ket{\Omega}\right)\left( \sum_{\vec{b}} \beta^*_{\vec{b}} \bra{\Omega} f_{b_{m}}\dots f_{b_1} \right) \left( \sum_{\vec{d}} \delta^*_{\vec{d}}\bra{\Omega}   f_{d_{s}}\dots f_{d_1}\right)\left( \sum_{\vec{c}}  \gamma_{\vec{c}} f^{\dagger}_{c_1}\dots f^{\dagger}_{c_{r}}\ket{\Omega}\right)= \ketbra{a}{b} \scp{d}{c} \\  \ketbra{a}{b}\wedge \ketbra{c}{d}=\sum_{\vec{a},\vec{b},\vec{c},\vec{d}} \alpha_{\vec{a}} \beta^*_{\vec{b}} \gamma_{\vec{c}} \delta^*_{\vec{d}}  (-1)^{nr+sm}  f^{\dagger}_{c_1}\dots f^{\dagger}_{c_{r}} f^{\dagger}_{a_1}\dots f^{\dagger}_{a_{n}} \proj{\Omega}  f_{b_{m}}\dots f_{b_1} f_{d_{s}}\dots f_{d_1}
\end{gather}
Since they follow SSR: $n\equiv m \mod 2$ and $r\equiv s \mod 2$. Therefore the parity of $nr+sm$ is always even. Thus it is found that: 
%\begin{table}[H]
%\centering
%\begin{tabular}{c|c|c}
%$n,m$ $\downarrow$ $r,s$ $\rightarrow$ & Even & Odd \\ \hline  Even & Even & Even \\ Odd & Even & Even 
%\end{tabular}
%\end{table}
\begin{gather}
\ketbra{a}{b}\wedge \ketbra{c}{d}=\sum_{\vec{a},\vec{b},\vec{c},\vec{d}} \alpha_{\vec{a}} \beta^*_{\vec{b}} \gamma_{\vec{c}} \delta^*_{\vec{d}}   f^{\dagger}_{c_1}\dots f^{\dagger}_{c_{r}} f^{\dagger}_{a_1}\dots f^{\dagger}_{a_{n}} \proj{\Omega}  f_{b_{m}}\dots f_{b_1} f_{d_{s}}\dots f_{d_1} \\ \tr_M\left(\ketbra{a}{b}\wedge \ketbra{c}{d}\right)= \sum_{\vec{a},\vec{b},\vec{c},\vec{d}} \alpha_{\vec{a}} \beta^*_{\vec{b}} \gamma_{\vec{c}} \delta^*_{\vec{d}}  \left( f^{\dagger}_{c_1}\dots f^{\dagger}_{c_{r}} \proj{\Omega} f_{d_{s}}\dots f_{d_1} \right)\left(\bra{\Omega} f_{b_{m}}\dots f_{b_1} f^{\dagger}_{a_1}\dots f^{\dagger}_{a_{n}} \ket{\Omega}\right)=\nonumber \\ =    \left( \sum_{\vec{c}} \gamma_{\vec{c}} f^{\dagger}_{c_1}\dots f^{\dagger}_{c_{r}} \ket{\Omega}\right)\left(\sum_{\vec{d}} \delta^*_{\vec{d}} \bra{\Omega} f_{d_{s}}\dots f_{d_1} \right)\left(\sum_{\vec{b}} \beta^*_{\vec{b}} \bra{\Omega} f_{b_{m}}\dots f_{b_1} \right)\left(\sum_{\vec{a}} \alpha_{\vec{a}} f^{\dagger}_{a_1}\dots f^{\dagger}_{a_{n}} \ket{\Omega}\right) = \ketbra{c}{d} \scp{b}{a}
\end{gather}
\end{proof}

Once we have shown the proofs for general linear SSR preserving operators, we move to prove the theorems that characterize classes of such SSR preserving linear operators. 

\begin{prop}\textbf{ Projectors})
\label{prop:Projector}.
Consider $\hat{P} \in \Gamma_S$ an SSR linear operator.  $\hat{P}$ is called a projector iff $\hat{P}^2=\hat{P}$ and $\hat{P}=\hat{P}^{\dagger}$, which is equivalent to have  the following form in the basis $\mathcal{B}'$:
\begin{align}
\hat{P}=P_{ee} \oplus P_{oo} = \left( \begin{array}{c|c}
P_{ee} & \hat{0} \\ \hline 
\hat{0} & P_{oo}
\end{array}\right)
\end{align}
where $P_{ee},P_{oo}\in M_{2^{N-1}\times 2^{N-1}}(\mathbb{C})$ are projectors of that space, and $\hat{0}$ is the zero element of $M_{2^{N-1}\times 2^{N-1}}(\mathbb{C})$. 
\end{prop}

\begin{proof}
By the Theorem \ref{thm:blockform} we know that any SSR linear operator $P$ decomposes in the basis $\mathcal{B}'$ as 
\begin{gather}
P= \left( \begin{array}{c|c}
P_{ee} & \hat{0} \\ \hline 
\hat{0} & P_{oo}
\end{array}\right) \quad \text{or} \quad P= \left( \begin{array}{c|c}
 \hat{0} & P_{oe} \\ \hline 
P_{eo} & \hat{0}
\end{array}\right)
\end{gather}
Now we only have to see that the anti-diagonal block form is not possible and then that $P$ is a projector iff $P_{ee}$ and $P_{oo}$ are.\\ 

If $P$ has an anti-diagonal block form then we obtain that $P^2$ is given by:

\begin{gather}
    P^2=\left( \begin{array}{c|c}
 P_{oe} \cdot P_{eo} & \hat{0} \\ \hline 
 \hat{0} & P_{eo} \cdot P_{oe}
\end{array}\right)
\end{gather}
which cannot be equal to $P$, unless we have the trivial case that can also be considered diagonal. Thus a projector cannot have an anti-diagonal block form.\\ 

So, considering only the diagonal block form we have that since
\begin{gather}
P^\dagger =\left( \begin{array}{c|c}
P_{ee}^\dagger & \hat{0} \\ \hline 
\hat{0} & P_{oo}^\dagger
\end{array}\right) \qquad \qquad P^2 =\left( \begin{array}{c|c}
P_{ee}^2 & \hat{0} \\ \hline 
\hat{0} & P_{oo}^2
\end{array}\right)
\end{gather}
Then, $P^\dagger=P$ iff $P_{ee}^\dagger=P_{ee}$ and $P_{oo}^\dagger=P_{oo}$, and  $P^2=P$ iff $P_{ee}^2=P_{ee}$ and $P_{oo}^2=P_{oo}$; proving the Proposition.
\end{proof}

\begin{prop}\textbf{ Observables)}
\label{prop:observables}
An operator $\hat{A}$ is Hermitian and is in $\Gamma_S$ iff $\hat{A}$ is an observable for a fermionic system under the SSR i.e. exists a set of orthonormal $\ket{\psi_i}\in \mathcal{H}_S$ such that $\hat{A}=\sum_i a_i \proj{\psi_i}$ with $a_i \in \mathbb{R}$.
\end{prop}

\begin{proof}
Lets start by naming the elements of the basis $\mathcal{B}'$ as $\ket{E_i}$ for the $i=1,\dots, 2^{N-1}$ firsts and $\ket{O_i}$ for the final $i=1,\dots, 2^{N-1}$ elements of the basis. Then since $A$ is an hermitian SSR operator it follows that
\begin{gather}
\left.A\right|_{\mathcal{B}'}=\begin{pmatrix}
A_E && 0 \\ 0 && A_O
\end{pmatrix}
\end{gather}
where $A_E$ and $A_O$ are hermitian matrices of dimension $2^{N-1}$. Thus they can be decomposed into real values with unitary matrices $U_E$ and $U_O$ respectively, and thus:
\begin{gather}
\left.A\right|_{\mathcal{B}'}=\begin{pmatrix}
A_E && 0 \\ 0 && A_O
\end{pmatrix}=\begin{pmatrix}
U_E D_E U_E^\dagger && 0 \\ 0 && U_O D_O U_O^\dagger
\end{pmatrix}=\begin{pmatrix}
U_E && 0 \\ 0 && U_O
\end{pmatrix}\cdot \begin{pmatrix}
D_E && 0 \\ 0 && D_O
\end{pmatrix}\cdot \begin{pmatrix}
U_E && 0 \\ 0 && U_O
\end{pmatrix}^\dagger
\end{gather}
thus since the matrix $\begin{pmatrix}
U_E && 0 \\ 0 && U_O
\end{pmatrix}$ is the matrix representation of a unitary operator $U$ (for more clarity see Theorem \ref{thm:Unitary}), we have that the states $\ket{E_i'}=U \ket{E_i}$ and $\ket{O_i'}=U\ket{O_i}$ will be eigenvectors with a real eigenvalue of $A$. It is only left to see that indeed the states satisfy the SSR. Is clear that the subspaces of even and odd spaces are invariant under the action of $U$. Thus the theorem is proven. 
\end{proof}

\begin{theorem*}\textbf{\ref{thm:Unitary}. Unitary)} 
$\hat{U}$ is an SSR unitary operator, acting on $\mathcal{H}_{S}$ if and only if the matrix representation of the operator $\hat{U}$ in the basis $\mathcal{B}'$ takes the following form:
\begin{align}
\hat{U}=U_{ee} \oplus U_{oo} = \left(\begin{array}{c|c}
U_{ee}  & \hat{0} \\ \hline \hat{0}  & U_{oo}  
\end{array}\right),
\end{align}
where $U_{ee}$ and $U_{oo}$ are unitary matrices, each in $M_{2^{N-1}\times 2^{N-1}}(\mathbb{C})$, living in the even and odd space respectively, and $\hat{0}$ is the zero element of $M_{2^{N-1}\times 2^{N-1}}(\mathbb{C})$.
\end{theorem*}

\begin{proof}
By the Theorem \ref{thm:blockform} we know that any SSR linear operator $U$ decomposes in the basis $\mathcal{B}'$ as 
\begin{gather}
U= \left( \begin{array}{c|c}
U_{ee} & \hat{0} \\ \hline 
\hat{0} & U_{oo}
\end{array}\right) \quad \text{or} \quad \left( \begin{array}{c|c}
\hat{0} & U_{oe} \\ \hline 
U_{eo} & \hat{0} 
\end{array}\right)
\end{gather}

First, we will prove that an anti-diagonal block unitary cannot exist. As mentioned in the main article, we do this by designing a protocol where the no-signalling principle is violated if such a unitary exists. The quantum circuit of the scheme is the following:

\begin{figure}[ht]
\centering
\includegraphics[width=0.40\textwidth]{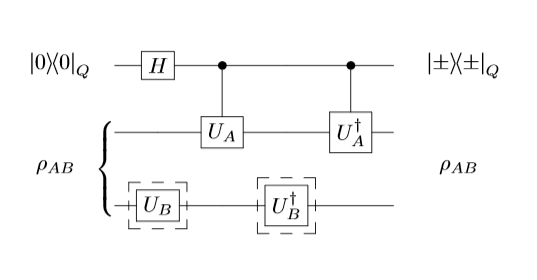}
\caption{ Scheme that shows the violation of no-signalling by anti-diagonal unitaries. The dashed boxes show the possibility that Bob applies those gates or not, depending on if he wants Alice to have a $\ket{+}$ or a $\ket{-}$ state.}
\end{figure} 

We assume that for an anti-diagonal unitary $\hat{U}$ in a set of modes $A$, we can apply the same unitary to another set of modes $B$. In this scheme, there are two separably distinct sets of fermionic modes $ A$ and $B$.  We denote the different applications as $\hat{U}_A$ and $\hat{U}_B$. In the spatial location of $A$, there is also a qubit system. Alice is able to couple the qubit system to the fermionic modes $A$ via a controlled-$\hat{U}_A$ operation, defined as: $\ketbra{0}{0} \otimes \mathbb{I}_A + \ketbra{1}{1} \otimes \hat{U}_A$.\\

The scheme consists in two cases, the case where $B$ decides to apply the anti-diagonal unitaries $\hat{U}_B$ and $\hat{U}_B^\dagger$ to their modes in the timing chosen in the scheme; and the case where the modes of $B$ do not get acted on. The initial state is $\ketbra{0}{0}\otimes \rho_{AB}$, where $\rho_{AB}$ is any SSR fermionic state. We can choose the qubit to be in this initial state. Now, lets split the scheme into its two parts. 
\begin{enumerate}
    \item $B$ is unmodified: Then we have that the protocol gives $\ketbra{+}{+}\otimes \rho_{AB}$. After the controlled-$\hat{U}_A$ gate we have $\frac{1}{2}\left(\ketbra{0}{0}\otimes \rho_{AB}+\ketbra{0}{1}\otimes \rho_{AB} \hat{U}_A^\dagger +\ketbra{1}{0}\otimes \hat{U}_A \rho_{AB} +\ketbra{1}{1}\otimes \hat{U}_A\rho_{AB}\hat{U}_A^\dagger\right)$ and after the final controlled-$\hat{U}_A^\dagger$ gate we have $\frac{1}{2}\left(\ketbra{0}{0}\otimes \rho_{AB}+\ketbra{0}{1}\otimes \rho_{AB} \hat{U}_A^\dagger\hat{U}_A +\ketbra{1}{0}\otimes \hat{U}_A^\dagger\hat{U}_A \rho_{AB} +\ketbra{1}{1}\otimes \hat{U}_A^\dagger\hat{U}_A\rho_{AB}\hat{U}_A^\dagger\hat{U}_A\right)=\ketbra{+}{+}\otimes\rho_{AB}$.
    \item $B$ applies the unitaries: from also $\ketbra{+}{+}\otimes \rho_{AB}$ we then have $\ketbra{+}{+}\otimes \hat{U}_B \rho_{AB} \hat{U}_B^\dagger$. After applying the controlled-$\hat{U}_A$ gate it is obtained $\frac{1}{2}\left(\ketbra{0}{0}\otimes \hat{U}_B \rho_{AB} \hat{U}_B^\dagger +\ketbra{0}{1}\otimes \hat{U}_B \rho_{AB} \hat{U}_B^\dagger \hat{U}_A^\dagger+\ketbra{1}{0}\otimes \hat{U}_A \hat{U}_B \rho_{AB} \hat{U}_B^\dagger +\ketbra{1}{1}\otimes \hat{U}_A \hat{U}_B \rho_{AB} \hat{U}_B^\dagger \hat{U}_A^\dagger\right)$. Then, when $\hat{U}_B^\dagger$ is applied we get $\frac{1}{2}\left(\ketbra{0}{0}\otimes \rho_{AB} +\ketbra{0}{1}\otimes \rho_{AB} \hat{U}_B^\dagger \hat{U}_A^\dagger \hat{U}_B+\ketbra{1}{0}\otimes \hat{U}_B^\dagger \hat{U}_A \hat{U}_B \rho_{AB} +\ketbra{1}{1}\otimes \hat{U}_B^\dagger \hat{U}_A \hat{U}_B \rho_{AB} \hat{U}_B^\dagger \hat{U}_A^\dagger \hat{U}_B\right)$. And after the final controlled-$\hat{U}_A^\dagger$ operation is applied, the result obtained is $\frac{1}{2}\left(\ketbra{0}{0}\otimes \rho_{AB} +\ketbra{0}{1}\otimes \rho_{AB} \hat{U}_B^\dagger \hat{U}_A^\dagger \hat{U}_B \hat{U}_A+\ketbra{1}{0}\otimes \hat{U}_A^\dagger \hat{U}_B^\dagger \hat{U}_A \hat{U}_B \rho_{AB} +\ketbra{1}{1}\otimes \hat{U}_A^\dagger \hat{U}_B^\dagger \hat{U}_A \hat{U}_B \rho_{AB} \hat{U}_B^\dagger \hat{U}_A^\dagger \hat{U}_B \hat{U}_A\right)$
\end{enumerate}

In order to proceed, is required that we proof the following statement. If $C$ and $D$ are anti-block diagonal SSR operators local on two non-overlapping sets of modes $A$ and $B$, then $CD=-DC$. To proof this statement, we just need to observe that such operators can be decomposed as a linear combination of monomials that are products of an odd number of creations and annihilation operators of the modes in $A$ and $B$ respectively. And since for each of this monomials $f_{\lambda_1}\dots f_{\lambda_n} f^\dagger_{\mu_1} \dots f^\dagger_{\mu_m}$ we observe that $(f_{\lambda_1}\dots f_{\lambda_n} f^\dagger_{\mu_1} \dots f^\dagger_{\mu_m})(f_{\nu_1}\dots f_{\nu_s} f^\dagger_{\eta_1} \dots f^\dagger_{\eta_r})=(-1)^{(n+m)(r+s)}(f_{\nu_1}\dots f_{\nu_s} f^\dagger_{\eta_1} \dots f^\dagger_{\eta_r}) (f_{\lambda_1}\dots f_{\lambda_n} f^\dagger_{\mu_1} \dots f^\dagger_{\mu_m}) $ if $\lambda_i,\mu_j \in A$ and $\nu_k,\eta_l \in B$ from this follows our crucial statement. We can deduce that the final state of the scheme for when the unitaries in $B$ are applied is $\frac{1}{2}\left(\ketbra{0}{0}\otimes \rho_{AB} +\ketbra{0}{1}\otimes \rho_{AB} \hat{U}_B^\dagger \hat{U}_A^\dagger \hat{U}_B \hat{U}_A+\ketbra{1}{0}\otimes \hat{U}_A^\dagger \hat{U}_B^\dagger \hat{U}_A \hat{U}_B \rho_{AB} +\ketbra{1}{1}\otimes \hat{U}_A^\dagger \hat{U}_B^\dagger \hat{U}_A \hat{U}_B \rho_{AB} \hat{U}_B^\dagger \hat{U}_A^\dagger \hat{U}_B \hat{U}_A\right) = \frac{1}{2}\left(\ketbra{0}{0}\otimes \rho_{AB} -\ketbra{0}{1}\otimes \rho_{AB} \hat{U}_A^\dagger \hat{U}_B^\dagger \hat{U}_B \hat{U}_A-\ketbra{1}{0}\otimes \hat{U}_B^\dagger \hat{U}_A^\dagger \hat{U}_A \hat{U}_B \rho_{AB} +\ketbra{1}{1}\otimes \hat{U}_B^\dagger \hat{U}_A^\dagger \hat{U}_A \hat{U}_B \rho_{AB} \hat{U}_A^\dagger \hat{U}_B^\dagger \hat{U}_B \hat{U}_A\right) = \frac{1}{2}\left(\ketbra{0}{0}\otimes \rho_{AB} -\ketbra{0}{1}\otimes \rho_{AB}-\ketbra{1}{0}\otimes \rho_{AB} +\ketbra{1}{1}\otimes \rho_{AB}\right)=\ketbra{-}{-}\otimes \rho_{AB}$. And since $\ket{+}$ and $\ket{-}$ are orthogonal states, Alice would know if Bob has applied the unitaries by measuring the qubit in the diagonal basis. Thus, Bob would have transmitted information to Alice without exchanging any particle nor any sort of classical communication. Bob via this protocol is able to transmit a message to Alice by acting remotely on his modes, no classical communication channel connects the two. Moreover, the information is transmitted instantaneously. Thus, the no-signalling principle would be violated. For this reason we conclude that anti-block diagonal unitaries cannot exist.\\    

Now, having discarded the anti-diagonal block case, we just have to see that $U$ is a unitary operator iff $U_{ee}$ and $U_{oo}$ are. And this follows directly from the block diagonal form action under hermitian conjugation and under product of block forms, e.g
\begin{gather}
U^\dagger =\left( \begin{array}{c|c}
U_{ee}^\dagger & \hat{0} \\ \hline 
\hat{0} & U_{oo}^\dagger
\end{array}\right) \qquad \qquad U U^\dagger =  \left( \begin{array}{c|c}
U_{ee} & \hat{0} \\ \hline 
\hat{0} & U_{oo}
\end{array}\right) \left( \begin{array}{c|c}
U_{ee}^\dagger & \hat{0} \\ \hline 
\hat{0} & U_{oo}^\dagger
\end{array}\right)= \left( \begin{array}{c|c}
U_{ee} U_{ee}^\dagger& \hat{0} \\ \hline 
\hat{0} & U_{oo} U_{oo}^\dagger
\end{array}\right)
\end{gather}
Thus, $U U^\dagger=\mathbb{I}_{\mathcal{H}}$ iff $U_{ee} U_{ee}^\dagger=\mathbb{I}_{2^{N-1}}$ and $U_{oo} U_{oo}^\dagger=\mathbb{I}_{2^{N-1}}$; proving the Theorem.
\end{proof}

Once we have proven the theorems that characterize the different types of SSR operators, we reproduce the proofs of the results that we need to discuss the notion of separable states properly. First, we prove the analogous Schmidt decomposition and purification procedures.

\begin{theorem*}\textbf{\ref{thm:schmidt}. Schmidt decomposition)} 
Given any bipartite, pure SSR fermionic state $\ket{\psi}_{AB} \in \mathcal{H}_{S}^{AB}$, there exist orthonormal basis $\{\ket{i}_A\} \in \mathcal{H}_{S}^A$ and $\{\ket{i}_B\} \in \mathcal{H}_{S}^B$, such that 
\begin{align}
\ket{\psi}_{AB}=\sum_i \sqrt{p_i} \ket{i}_A\wedge \ket{i}_B, 
\end{align}
where $\{p_i\}$ are probabilities.
\end{theorem*}

\begin{proof}
First of all,  the state $\ket{\psi}$ can be decomposed on the canonical basis $\mathcal{B}$ where its elements can be thought as products of the canonical basis of the subsystems:
\begin{equation}
\ket{\psi}=\sum_{i,j} \lambda_{i,j} \ket{e_i}\wedge \ket{e_j}
\end{equation}
where $\{\ket{e_i}\}$ is a basis of $\mathcal{H}_A$ and $\{\ket{e_j}\}$ is the canonical basis of $\mathcal{H}_B$. Therefore this expression can be transformed into another, transforming the elements of $\mathcal{B}_A$ into another basis $\mathcal{B}'_A$ and transforming the elements of $\mathcal{B}_B$ into another basis $\tilde{\mathcal{B}}_B$:
\begin{equation}
\ket{\psi}=\sum_{i,j} \mu_{i,j} \ket{f_i}\wedge \ket{g_j}
\end{equation}
with $\{f_i\}\in\mathcal{B}'_A$ and $\{g_j\}\in\tilde{\mathcal{B}}_B$. Since the transformation is unitary, the state stays well normalized. Therefore for any basis on $\mathcal{H}_A$ and $\mathcal{H}_B$ the state can be decomposed in these basis. \\

A new basis of $\mathcal{H}_B$ can be defined if the terms are grouped:
\begin{gather}
\ket{\psi}=\sum_{i} \ket{f_i}\wedge(\sum_j \mu_{i,j} \ket{g_j})=\sum_{i} \ket{f_i}\wedge\ket{h_i}
\end{gather}
where $\{\ket{h_i}\}$ it is not normalized neither orthogonal. \\

Once this description has been done, as a basis for $\mathcal{H}_A$ lets choose the basis in which $\rho_A$ is diagonal. $\rho_A$ is obtained by partial tracing B in $\rho=\proj{\psi}$. Therefore since we have $\rho_A=\sum_i p_i \proj{i}$, and it has been chosen $\ket{f_i}=\ket{i}$ lets see the relation with what we previously had:
\begin{gather}
\rho=\sum_{i,j} \ket{i}\wedge \ket{h_i}\bra{j}\wedge \bra{h_j} \Rightarrow \rho_A=\sum_{i,j} \ketbra{i}{j} \scp{h_j}{h_i}=\sum_i p_i \proj{i}
\end{gather}
where in the first implication it has been used the Proposition \ref{prop:pureptssr}. This relations imply that $\scp{h_j}{h_i}=p_i\delta_{i,j}$. Therefore the $\{h_i\}$ are indeed orthogonal. Defining $\ket{\tilde{i}}\equiv \frac{\ket{h_i}}{\sqrt{p_i}}$ the set $\{\ket{\tilde{i}}\}$ conform an orthonormal basis of $\mathcal{H}_B$, and therefore it can be written:

\begin{equation}
\ket{\psi}=\sum_i \sqrt{p_i} \ket{i}\wedge \ket{\tilde{i}}
\end{equation}
\end{proof}

\begin{corollary*}\textbf{\ref{cor:purif}. Purification)}
If $\rho \in \mathcal{R}_{S}^M$, then there exists a fermionic space $E$ of $M$ modes and a pure state $\omega \in \mathcal{H}_{S}^M  \wedge E$, such that $\tr_E(\omega)=\rho$. 
\end{corollary*}

\begin{proof}
We know that we can decompose any $\rho$ as $\rho=\sum_i p_i \proj{\psi_i}$ with $\ket{\psi_i}$ chosen to be SSR states. Since the sum is finite, we consider a set of $M$ new modes, where $2^M$ is the number of summing terms in the decomposition of $\rho$. With this set, we generate a new fermionic space $E$ with $M$ modes. Now we choose the state:
\begin{gather}
\omega=\sum_i \sqrt{p_i} \ket{\psi_i}\wedge \ket{i}
\end{gather}
where $\ket{i}$ are states of the canonical basis of our new space of $M$ modes. Their parity is the same as the parity of $\ket{\psi_i}$. So the global state $\omega$ has an even parity in $\mathcal{H}_{S}\wedge E$. Now, it is straightforward to check with the proved properties of the partial trace that $\tr_E(\omega)=\rho$.
\end{proof}

Once we have these tools, we prove the final results that characterize fermionic SSR uncorrelated states. We use them to discuss the relationship between the three different definitions presented in the main article. 

\begin{prop}
\label{prop:ProdState}
For a SSR fermionic bipartite state $\rho_{AB} \in \mathcal{R}_{S}^{AB}$: 
\begin{gather}
 \tr(\rho_{AB} (\hat{O}_A\wedge \hat{O}_B))=\tr(\rho_A \hat{O}_A)\tr(\rho_B \hat{O}_B) \quad \forall \hat{O}_X\in \mathcal{O}_X \quad \Longleftrightarrow \quad  \rho_{AB}=\rho_A \wedge \rho_B 
\end{gather}
where $\mathcal{O}_X$ is the set of all Hermitian operators local on the subspace $\mathcal{H}^X$
\end{prop}
\begin{proof}
$\Leftarrow$: If we have $\rho_{AB}=\rho_A \wedge \rho_B$, then given any $\hat{O}_A\in \mathcal{O}_A$ and any $\hat{O}_B\in \mathcal{O}_B$, $\tr(\rho_{AB}(\hat{O}_A \wedge \hat{O}_B))=\tr((\rho_A \wedge \rho_B)(\hat{O}_A \wedge \hat{O}_B))$ which by Lemma \ref{lema:distribu} is equal to $\tr((\rho_A \hat{O}_A)\wedge (\rho_B \hat{O}_B) )$ and then by Lemma \ref{lema:trprod} gives $\tr(\rho_A \hat{O}_A)\cdot \tr(\rho_B \hat{O}_B)$ just as desired. Since the equality holds for any Hermitian local operators, holds for all SSR local Hermitian operators and all local Hermitian operators. \\

$\Rightarrow$: In order to proof the other implication we use that since we know that the $\Leftarrow$ implication holds, for all  $\hat{O}_A \in \mathcal{O}_A$ and for all $ \hat{O}_B \in \mathcal{O}_B$: $\tr(\rho_A \hat{O}_A)\cdot\tr(\rho_B \hat{O}_B)=\tr((\rho_A\wedge \rho_B)(\hat{O}_A\wedge \hat{O}_B))$ So the condition  $ \tr(\rho_{AB} (\hat{O}_A\wedge \hat{O}_B))=\tr(\rho_A \hat{O}_A)\tr(\rho_B \hat{O}_B)$ for all $\hat{O}_X\in \mathcal{O}_X$ is equivalent to $ \tr(\rho_{AB} (\hat{O}_A\wedge \hat{O}_B))=\tr((\rho_A\wedge \rho_B)(\hat{O}_A\wedge \hat{O}_B))$ for all $\hat{O}_X\in \mathcal{O}_X$  , which is equivalent by linearity of the trace and distribution of the product to $ \tr((\rho_{AB}-\rho_A\wedge \rho_B) (\hat{O}_A\wedge \hat{O}_B))=0$ for all $\hat{O}_X\in \mathcal{O}_X$. And by defining $D= \rho_{AB}-\rho_A\wedge \rho_B$ the statement that we want to proof is equivalent to proof that: if  $\tr(D (\hat{O}_A\wedge \hat{O}_B))=0$ for all $\hat{O}_X\in \mathcal{O}_X$ then $D=0$. And since $D$ is a Hermitian operator because is a lineal combination of Hermitian operators, we can apply Lemma \ref{lema:tracezero} which is exactly this statement, so we have proven the implication. 
\end{proof}

So we have proven that the first and second definitions of uncorrelated states are equivalent. Note though, that in the proof of the Proposition \ref{prop:ProdState} we see that in the implication $\Leftarrow$ also relies the proof that the states that satisfy the second possible definition of uncorrelated state mentioned in the article will also satisfy the third. This is since all SSR local observables are local Hermitian SSR operators, that are a subset of local Hermitian operators. But, as mentioned in \cite{Banuls09} there exist states that adhere to the third definition but not to the second one. One counterexample is, for a 2 mode fermionic system, under the matrix representation in the basis $\mathcal{B}$ 
\begin{gather}
\rho_{12}=\frac{1}{16}\begin{pmatrix} 9 & 0 & 0 & -i \\ 0 & 3 & -i & 0 \\ 0 & i & 3 & 0 \\ i & 0 & 0 & 1 \end{pmatrix} \quad \rho_1=\frac{1}{16}\begin{pmatrix} 12 & 0 \\ 0 & 4 \end{pmatrix} \quad \rho_2= \frac{1}{16}\begin{pmatrix} 12 & 0 \\ 0 & 4 \end{pmatrix}
\end{gather} 
The result relies on the fact that there are only 2 linearly independent Hermitian SSR local operators in the 1 mode system, and both have a diagonal representation in the basis $\mathcal{B}$, being $\{\mathbb{I}, f_1 f_1^\dagger\}$. Now, assuming that $\rho_{AB}$ is pure, we are able to proof that the three proposed definitions of uncorrelated states are the same. 

\begin{proof}
We have to prove the implication from the third to the second definition, and we do it by \textit{reductio ad absurdum}. We start writing $\rho_{AB}$ as $\rho_{AB}=\proj{\psi}$ since we know that is pure. Now it will be assumed that $\rho_{AB}\neq \rho_A \wedge \rho_B$, and get to a contradiction:\\

Using the Schmidt decomposition from Theorem \ref{thm:schmidt} proven in this appendix, one can decompose $\rho_{AB}$ as:

\begin{gather}
\rho_{AB}=\sum_{i,j} \sqrt{p_i p_j} \ketbra{i_A
}{j_A} \wedge \ketbra{i_B}{j_B}
\end{gather}

We can say that if $\rho_{AB} \neq \rho_A \wedge \rho_B$ then the number of non-zero $p_i$ is greater than 1. Therefore, without loss of generality we can consider $p_1,p_2 \in (0,1)$. Now, lets see what is obtained:

\begin{gather}
\tr((P_A\wedge P_B)\rho)=\sum_{i,j
} \sqrt{p_i p_j} \tr((P_A\wedge P_B) \ketbra{i_A}{j_A} \wedge \ketbra{i_B}{j_B}
)=\nonumber \\ =\sum_{i,j
} \sqrt{p_i p_j} \tr(P_A \ketbra{i_A}{j_A} \wedge P_B \ketbra{i_B}{j_B}
)= \sum_{i,j
} \sqrt{p_i p_j} \tr(P_A \ketbra{i_A}{j_A}) \tr(P_B \ketbra{i_B}{j_B}
) 
\end{gather}

Now, it is easy to calculate the corresponding partial traces to obtain $\rho_A=\sum_{i} p_i \proj{i_A}$ and $\rho_B=\sum_{j} p_j \proj{j_B}$. Thus, the right hand part of the first uncorrelation relation becomes:

\begin{gather}
\tr(P_A \rho_A)\tr(P_B \rho_B)=\sum_{i,j} p_i p_j \tr(P_A \proj{i_A}) \tr(P_B \proj{j_B})
\end{gather}

Now we will see that these two quantities cannot be equal. If they were, for all $P_A,P_B$, choosing $P_A=\proj{1_A}$ and $P_B=\proj{2_B}$, which obviously are hermitian operators, one will obtain:

\begin{gather}
\tr((P_A\wedge P_B) \rho_{AB})=\sum_{i,j} \delta_{i,1} \delta_{i,2} \sqrt{p_i p_j} \tr(\ketbra{1_A}{j_A}) \tr(\ketbra{2_B}{j_B})=0\\ \tr(P_A\rho_A)\tr(P_B \rho_B)=\sum_{i,j} p_i p_j \delta_{1,i} \delta_{2,j} \tr(\ketbra{1_A}{i_A}) \tr(\ketbra{2_B}{j_B})=\sum_{i,j} p_i p_j \delta_{1,i}^2 \delta_{2,j}^2=p_1 p_2
\end{gather}

So, we would obtain that $0=p_1 p_2$, but since we have that $p_1,p_2 \in (0,1)$ we can say that we arrived to a contradiction. Thus if $\rho\neq \rho_A \wedge \rho_B$ then $\tr((P_A\wedge P_B)\rho)\neq \tr(P_A\rho_A)\tr(P_B \rho_A)$ for all $P_A,P_B$ local hermitian operators on $\mathcal{H}_A$ and $\mathcal{H}_B$ respectively.
\end{proof}

Thus, all three definitions of uncorrelated states agree for pure states.

%\newpage

\section{Proofs of CPTP-Kraus-Stinespring equivalences} 
\label{sec:prooftheorem}

In this Appendix, we present the complete proofs of Theorem \ref{thm:Sinespring-SupO-CPTP} and the general characterization of quantum operations.

To then proof more easily Theorem \ref{thm:Sinespring-SupO-CPTP}, we first prove the equivalence that holds for general quantum operations. 

\begin{theorem}
\label{thm:Kraus-CPTP}
\textbf{General quantum operations)} 
For a SSR fermionic quantum operation represented by a map $\varphi: \mathcal{R}_{S}^{N} \rightarrow \Gamma_{S}^{M}$ that transforms $N$-mode SSR fermionic states in $\mathcal{R}_{S}^{N}$ to a $M$-mode SSR fermionic operators in $\Gamma_{S}^{M}$, the following statements are equivalent.

\begin{enumerate}
\item (Operator-sum representation.) There exists a set of linear operators $E_k: \mathcal{H}_{S}^{N} \rightarrow \mathbb{C}\mathcal{H}_{S}^{M}$, with $0 \leq \sum_k E_k^{\dagger}E_k \leq \mathbb{I}_{N}$, such that:
\begin{equation}
\varphi(\rho)=\sum_k E_k  \rho E_k^{\dagger}
\end{equation}

\item (Axiomatic formalism.) $\varphi$ fulfills the following properties:
\begin{itemize}
\item $ \tr(\varphi(\rho))$ is a probability, i.e. $0\leq \tr(\varphi(\rho))\leq 1$ for all $\rho\in \mathcal{R}^N_S$.
\item Convex-linear, i.e. $\varphi\left(\sum_i p_i\rho_i \right)=\sum_i p_i \varphi(\rho_i)$ with $p_i$ probabilities and $\rho\in \mathcal{R}^N_S$.
\item $\varphi: \mathcal{R}_{S}^{N} \rightarrow \Gamma_{S}^{M}$ is CP.   
\end{itemize}
\end{enumerate}
\end{theorem}

\begin{proof}
To proof this equivalence, it will be seen first that 1 implies 2 and latter that 2 implies 1. For 1 implies 2 it has to be seen that if $\varphi(\rho)=\sum_k E_k^{\dagger}\rho E_k$ with the stated properties then the set of axioms is fulfilled. It is clear that the property b) is fulfilled due to the linear properties of operator sum and the distributive property. To proof a) is easy to see that follows from the cyclic property of the trace and the preservation of inequalities by the trace operator and the multiplication by a positive operator. To proof c) we first check that indeed $\varphi: \mathcal{R}_{S}^{N} \rightarrow \Gamma_{S}^{M}$. This can easily seen that follows due to the SSR preservation property of the Kraus operators $E_k$. Now, to proof that is CP, assume that $K\in \mathbb{N}$ and that $L$ is a fermionic Hilbert space of $K$ modes. Then choosing any state $\ket{\psi}\in \mathcal{H}_S^M\wedge L$ we can define for every $E_i$ an unnormalized state $\ket{\phi_i}\equiv (E_i^{\dagger}\wedge \mathbb{I}_K)\ket{\psi}$ where it can be easily checked that $0\leq \scp{\phi_i}{\phi_i}\leq 1$. Now if $A$ is an arbitrary positive operator of the Hilbert space $\mathcal{H}_S^N\wedge L$ we can see that:
\begin{gather}
\bra{\psi}(E_i\wedge\mathbb{I}_K)A(E_i^{\dagger}\wedge\mathbb{I}_K)\ket{\psi}=\bra{\phi_i}A\ket{\phi_i}\geq 0
\end{gather}
where the last step is true for the positivity of $A$ and the norm of $\ket{\phi_i}$. Once we have that, since \begin{gather}
    (\varphi\wedge\mathbb{I}_K)(A)=\sum_i (E_i\wedge \mathbb{I}_K)A(E_i^{\dagger}\wedge \mathbb{I}_K)
\end{gather} 
and the numerable sum of positive elements is positive, it is found that for an arbitrary $K$, an arbitrary state $\ket{\psi}\in\mathcal{H}_S^M\wedge L$ and for an arbitrary positive operator $A$ for $\mathcal{H}_S^N\wedge L$,  $\bra{\psi}(\varphi\wedge \mathbb{I}_K )(A) \ket{\psi}\geq 0$, and therefore  $(\varphi\wedge \mathbb{I}_K )(A)\geq 0$ and thus $\varphi$ is CP. \\

For the opposed implication, we have the map $\varphi$ fulfilling the three axioms. Now let us consider an additional fermionic Hilbert space of $N$ modes, that we will call $L$ and consider the global Hilbert space $\mathcal{H}_S^N\wedge L$. In these Hilbert spaces, orthonormal basis $\{\ket{i_H}\}\in \mathcal{H}_S^N$ and $\{\ket{i_L}\}\in L$ indexed by the same numerable label $i=1,\dots,2^N$ can be chosen. We can choose this basis so that the first $2^{N-1}$ are even SSR states and the last $2^{N-1}$ are odd SSR states, for both spaces.\\

Now we consider the even SSR state for the global system:
\begin{gather}
\ket{\alpha}\equiv\frac{1}{2^N} \sum_{i=1}^{2^N} \ket{i_H}\wedge \ket{i_L}
\end{gather}

Now from this definition, an operator on the global Hilbert space is defined:
\begin{gather}
\sigma \equiv (\varphi\wedge \mathbb{I}_N)(\proj{\alpha})
\end{gather}

Once made this construction, it is known that any SSR state $\ket{\eta}\in \mathcal{H}_S^N$ can be written as $\ket{\eta}=\sum_{j=1}^{2^{N-1}} \eta_j \ket{j_H}$ or $\ket{\eta}=\sum_{j=1+2^{N-1}}^{2^{N}} \eta_j \ket{j_H}$. To each SSR state, an analogue in $L$ is considered:
\begin{gather}
    \ket{\tilde{\eta}}=\sum_{j=1}^{2^{N-1}} \eta_j^* \ket{j_L} \quad \text{or}\quad \ket{\tilde{\eta}}=\sum_{j=1+2^{N-1}}^{2^N} \eta_j^* \ket{j_L}
\end{gather} 
which is also a SSR state on $L$ and with the same parity. For the properties of the wedge product and the definition of $\sigma$ it is found that:
\begin{gather}
\sigma=\frac{1}{2^{2N}}\sum_{i,j=1}^{2^N} \varphi(\ketbra{i_H}{j_H})\wedge \ketbra{i_L}{j_L}
\end{gather}
Under the SSR, it is found that either:
\begin{gather}
\bra{\tilde{\eta}}\sigma\ket{\tilde{\eta}}=\frac{1}{2^{2N}}\sum_{i,j=1}^{2^{N-1}} \varphi(\ketbra{i_H}{j_H})\eta_{j}^*\eta_i \quad \text{or} \quad \bra{\tilde{\eta}}\sigma\ket{\tilde{\eta}}=\frac{1}{2^{2N}}\sum_{i,j=2^{N-1}+1}^{2^{N}} \varphi(\ketbra{i_H}{j_H})\eta_{j}^*\eta_i
\end{gather}
Now applying the axiom 2 it is found $\bra{\tilde{\eta}}\sigma\ket{\tilde{\eta}}\cdot 2^{2N}=\varphi(\ketbra{\eta}{\eta})$. So from $\sigma$ it can be recovered $\varphi$.\\

Now, since $\varphi$ for the axiom 3 is CP, in particular, we find that $\sigma$ must be a positive operator. This fact implies that we can have a diagonal decomposition of the form:
\begin{gather}
\sigma=\sum_{i=1}^{2^{M+N}} a_i \proj{s_i}
\end{gather}
where $a_i\geq 0$. Now for the axiom 3, it implies that $\varphi\wedge \mathbb{I}_N$ is an operator that preserves the parity SSR for states in $\mathcal{H}_S^N\wedge L$. Since $\proj{\alpha}$ is an SSR state, $\sigma$ is also an SSR operator. This fact implies that we can choose $\ket{s_i}$ so that they are SSR states, and order them so that the first $2^{M+N-1}$ are even, and the rest odd.\\

Once this has been seen, we define the operators $E_k$ using the mentioned decomposition above. Given a SSR state of $\mathcal{H}_S^N$, the action of $E_k$ is defined as: $E_k(\ket{\eta})=2^N \sqrt{a_k}\bra{\tilde{\eta}}\ket{s_k}$.\\

Now, given a $\rho$ of $\mathcal{H}_S^N$ it can be written as $\rho=\sum_i p_i \proj{\eta_i}$. Then we have that:
\begin{gather}
\varphi(\rho)= \sum_i p_i \varphi(\proj{\eta_i})=\sum_i p_i 2^{2N}\bra{\tilde{\eta_i}}\sigma\ket{\tilde{\eta_i}}=\sum_i p_i 2^{2N} \bra{\tilde{\eta_i}}\left(\sum_{k=1}^{2^{M+N}} a_k\proj{s_k}\right)\ket{\tilde{\eta_i}}=\nonumber \\ =\sum_i p_i \sum_{k=1}^{2^{M+N}} E_k \ket{\eta_i}\bra{\eta_i}E_k^{\dagger}=\sum_k E_k\rho E_k^{\dagger}
\end{gather}

The property that $0\leq \sum_k E_k^{\dagger}E_k\leq \mathbb{I}_N$ follows directly from the axiom 1 and the last statement. So is just left to see that $E_k$ preserves the parity SSR structure. Since for any $\ket{\eta}\in\mathcal{H}_S^N$, $E_k$ acts on it as 
\begin{gather}
    E_k \ket{\eta}=2^N \sqrt{a_k}\bra{\tilde{\eta}}\ket{s_k}
\end{gather} 
Since we have seen that $\ket{s_k}$ is super selected on the global space, it can be decomposed as:
\begin{gather}
    \ket{s_k}=\sum_{j=1}^{2^{M-1}} \sum_{l=1}^{2^{N-1}} b_{k,j,l}\ket{j_H}\wedge\ket{l_L}+\sum_{j=1+2^{M-1}}^{2^{M}}\sum_{l=1+2^{N-1}}^{2^{N}} c_{k,j,l}\ket{j_H}\wedge\ket{l_L} \medspace \text{or} \\
    \ket{s_k}=\sum_{j=1}^{2^{M-1}} \sum_{l=1+2^{N-1}}^{2^{N}} d_{k,j,l}\ket{j_H}\wedge\ket{l_L}+\sum_{j=1+2^{M-1}}^{2^{M}}\sum_{l=1}^{2^{N-1}} e_{k,j,l}\ket{j_H}\wedge\ket{l_L}
\end{gather}

Then, if $\ket{\eta}$ is even,  $\ket{\tilde{\eta}}$ is even and it is found:
\begin{gather}
\bra{\tilde{\eta}}\ket{s_k}=\sum_{j=1}^{2^{M-1}} \sum_{l=1}^{2^{N-1}} b_{k,j,l}\ket{j_H}\bra{\tilde{\eta}}\ket{l_L}=\sum_{j=1}^{2^{M-1}} \tilde{b}_{k,j} \ket{j_H}\in \mathbb{C}\mathcal{H}_S^M \\ \text{or}\quad \bra{\tilde{\eta}}\ket{s_k}=\sum_{j=1+2^{M-1}}^{2^{M}}\sum_{l=1}^{2^{N-1}} c_{k,j,l}\ket{j_H}\bra{\tilde{\eta}}\ket{l_L}=\sum_{j=1+2^{M-1}}^{2^{M}} \tilde{e}_{k,j} \ket{j_H} \in \mathbb{C}\mathcal{H}_S^M
\end{gather}
And if $\ket{\eta}$ is odd, using the same arguments it is found:
\begin{gather}
\bra{\tilde{\eta}}\ket{s_k}=\sum_{j=1+2^{M-1}}^{2^{M}}\sum_{l=1+2^{N-1}}^{2^N} c_{k,j,l}\ket{j_H}\bra{\tilde{\eta}}\ket{l_L}= \sum_{j=1+2^{M-1}}^{2^{M}}\tilde{c}_{k,j} \ket{j_H}\in \mathbb{C}\mathcal{H}_S^M \\ \text{or}\quad \bra{\tilde{\eta}}\ket{s_k}= \sum_{j=1}^{2^{M-1}} \sum_{l=1+2^{N-1}}^{2^N} d_{k,j,l}\ket{j_H}\bra{\tilde{\eta}}\ket{l_L}=\sum_{j=1}^{2^{M-1}} \tilde{d}_{k,j} \ket{j_H} \in \mathbb{C}\mathcal{H}_S^M
\end{gather}
Therefore all the properties are checked, and the implication is proven.
\end{proof}

\begin{theorem*}
\textbf{\ref{thm:Sinespring-SupO-CPTP}. General quantum channels)} 
For a SSR fermionic quantum channel represented by a map $\varphi: \mathcal{R}_{S}^{N} \rightarrow \mathcal{R}_{S}^{N}$  the following statements are equivalent.

\begin{enumerate}
\item (Operator-sum representation.) There exists a set of SSR linear operators $E_k: \mathcal{H}_{S}^{N} \rightarrow \mathbb{C}\mathcal{H}_{S}^{N}$, where $\sum_k E_k^{\dagger}E_k = \mathbb{I}_{N}$, such that:
\begin{equation}
\varphi(\rho)=\sum_k E_k  \rho E_k^{\dagger}
\end{equation}

\item (Axiomatic formalism.) $\varphi$ fulfills the following properties:
\begin{itemize}
\item Is trace preserving, i.e. $\tr(\varphi(\rho))= 1$ for all $\tr(\rho)= 1$ and $\rho\in\mathcal{R}_{S}^{N}$.
\item Convex-linear, i.e. $\varphi\left(\sum_i p_i\rho_i \right)=\sum_i p_i \varphi(\rho_i)$ with $p_i$ probabilities.
\item $\varphi: \mathcal{R}_{S}^{N} \rightarrow \mathcal{R}_{S}^{N}$ is CP.  
\end{itemize}
\item (Stinespring dilation.) There exists a fermionic $K$-mode environment ($L$) with Hilbert space $L=\mathcal{H}_{S}^{K}$ and $K\geq N$, a SSR pure state $\omega=\proj{\psi} \in L$ and a parity SSR respecting unitary operator $\hat{U}$ that acts on $\mathcal{H}_{S}^{N}\wedge L$, such that:
\begin{equation}
\varphi(\rho)=\tr_{L}(\hat{U}(\rho\wedge \omega)\hat{U}^{\dagger}), \qquad \forall \rho\in \mathcal{R}_{S}^{N}.
\end{equation}
\end{enumerate}
\end{theorem*}

\begin{proof}
First, we will proof the implication: 3 implies 1.\\

We choose an orthonormal SSR basis on $E$ denoted by $(f_i)$ such that the first $2^{K-1}$ modes are even, and the last $2^{K-1}$ are odd. Then we have that, since the SSR is respected by all the operators the partial trace can be written as 
\begin{gather}
\varphi(\rho)=\tr_L (U(\rho\wedge\proj{\psi})U^{\dagger})=\sum_{i}  \bra{f_i}_{L} U(\rho \wedge\proj{\psi})U^{\dagger})\ket{f_i}_L
\end{gather}
where $\ket{f_i}_L$ is an element that only acts on the subspace $L$ of $\mathcal{H}_S^N\wedge L$; it can be seen as $\mathbb{I}_H\wedge \ket{f_i}$. This equality holds due to properties of the partial trace for SSR operators, exposed in Subsection \ref{subsec:pt} . \\

Using this terminology and the fact that since $\rho$ is a SSR state, then it can be seen that :
\begin{gather}
\rho \wedge \proj{\psi}=\ket{\psi}_L \rho \bra{\psi}_L
\end{gather}
 And therefore we have:
 
 \begin{gather}
 \varphi(\rho)= \sum_{i} \bra{f_i}_{L} U(\rho \wedge\proj{\psi})U^{\dagger})\ket{f_i}_L=\sum_{i} \bra{f_i}_{L} U(\ket{\psi}_L \rho \bra{\psi}_L)U^{\dagger})\ket{f_i}_L=\nonumber \\ = \sum_{i} (\bra{f_i}_{L} U\ket{\psi}_L) \rho (\bra{\psi}_L U^{\dagger}\ket{f_i}_L)
 \end{gather}
Now if we define $E_{i}=\bra{f_i}_{L} U\ket{\psi}_L$, it is indeed a linear operator of the space $\mathcal{H}_S^N$. And it is found that:
\begin{gather}
E_{i}^{\dagger}=\bra{\psi}_L U^{\dagger}\ket{f_i}_L
\end{gather}
Therefore, it is proved that $\varphi(\rho)=\sum_{i} E_{i} \rho E_{i}^{\dagger}$. Now the two other properties of the operators $E_j$ have to be seen. First, if we compute:
\begin{gather}
\sum_j E_j^{\dagger}E_j=\sum_{j} \bra{\psi}_L U^{\dagger}\ket{f_j}_L \bra{f_j}_{L} U\ket{\psi}_L=\nonumber \\ =\bra{\psi}_L U^{\dagger}\sum_j \ket{f_j}_L \bra{f_j}_{L} U\ket{\psi}_L =\nonumber \\=\bra{\psi}_L U^{\dagger}U\ket{\psi}_L=\bra{\psi}_L \mathbb{I}_N\wedge \mathbb{I}_K\ket{\psi}_L=\mathbb{I}_N
\end{gather}
the last property that has to be checked is that $\forall j$ if $\ket{\eta}$ is a SSR state then $E_j\ket{\eta}$ is also a SSR state. Therefore lets suppose that $\eta$ is a SSR state. Now since $U$ is a unitary that preserves SSR, it can only take the diagonal form, due to Theorem \ref{thm:Unitary}.\\

In order to make the notation lighter lets assume that the ordering is clear, and that if it is denoted a state by $\ket{E_i}$ it means that is an even state of the corresponding space, and if $\ket{O_j}$ then it is odd, and the sets where they belong conform an orthonormal basis. Then any SSR unitary acting in a space of $N+K$ modes can be decomposed as:
\begin{gather}
\hat{U}=\sum_{i,j,k,l} a_{i,j,k,l} \ket{E_i}\wedge \ket{E_j}\bra{E_k}\wedge \bra{E_l}+ b_{i,j,k,l} \ket{E_i}\wedge \ket{E_j}\bra{O_k}\wedge \bra{O_l}+\nonumber \\+c_{i,j,k,l}\ket{O_i}\wedge \ket{O_j}\bra{E_k}\wedge \bra{E_l}+d_{i,j,k,l}\ket{O_i}\wedge \ket{O_j}\bra{O_k}\wedge \bra{O_l}+\nonumber \\+A_{i,j,k,l} \ket{E_i}\wedge \ket{O_j}\bra{E_k}\wedge \bra{O_l}+ B_{i,j,k,l} \ket{E_i}\wedge \ket{O_j}\bra{O_k}\wedge \bra{E_l}+\nonumber \\+C_{i,j,k,l}\ket{O_i}\wedge \ket{E_j}\bra{E_k}\wedge \bra{O_l}+D_{i,j,k,l}\ket{O_i}\wedge \ket{E_j}\bra{O_k}\wedge \bra{E_l} 
\end{gather}

The initial SSR environment state $\ket{\psi}$ can either be even or odd. Moreover, for every $E_j$, the corresponding $\ket{f_j}$ can also be even or odd. Therefore there are four different cases to be taken into account. 

\begin{gather}
\ket{\psi} \in \text{even   :} \nonumber \\ 
U \ket{\psi}_L=\sum_{i,j,k,l} a_{i,j,k,l} \ket{E_i}\wedge \ket{E_j}\bra{E_k} \scp{E_l}{\psi}+ c_{i,j,k,l} \ket{O_i}\wedge \ket{O_j}\bra{E_k} \scp{E_l}{\psi}+ \nonumber \\ +B_{i,j,k,l} \ket{E_i}\wedge \ket{O_j}\bra{O_k} \scp{E_l}{\psi}+D_{i,j,k,l} \ket{O_i}\wedge \ket{E_j}\bra{O_k} \scp{E_l}{\psi}  \\
\ket{\psi} \in \text{odd  :} \nonumber \\ 
U \ket{\psi}_L=\sum_{i,j,k,l} b_{i,j,k,l} \ket{E_i}\wedge \ket{E_j}\bra{O_k} \scp{O_l}{\psi}+ d_{i,j,k,l} \ket{O_i}\wedge \ket{O_j}\bra{O_k} \scp{O_l}{\psi}+ \nonumber \\ + A_{i,j,k,l} \ket{E_i}\wedge \ket{O_j}\bra{E_k} \scp{O_l}{\psi}+C_{i,j,k,l} \ket{O_i}\wedge \ket{E_j}\bra{E_k} \scp{O_l}{\psi}
\end{gather}
Thus, the four combinations end up giving:
\begin{gather}
\ket{\psi} \in \text{even ,} \medspace \ket{f_{i'}} \in \text{even :} \quad \qquad \bra{f_{i'}}_L U \ket{\psi}_L=\sum_{i,j,k,l} a_{i,j,k,l} \scp{f_{i'}}{E_j} \scp{E_l}{\psi} \ketbra{E_i}{E_k} +D_{i,j,k,l} \scp{f_{i'}}{E_j} \scp{E_l}{\psi} \ketbra{O_i}{O_k} \\
\ket{\psi} \in \text{even ,} \medspace \ket{f_{i'}} \in \text{odd :} \quad \qquad \bra{f_{i'}}_L U \ket{\psi}_L=\sum_{i,j,k,l} c_{i,j,k,l} \scp{f_{i'}}{O_j} \scp{E_l}{\psi} \ketbra{O_i}{E_k} +B_{i,j,k,l} \scp{f_{i'}}{O_j} \scp{E_l}{\psi} \ketbra{E_i}{O_k}\\
\ket{\psi} \in \text{odd ,} \medspace \ket{f_{i'}} \in \text{even :} \quad \qquad \bra{f_{i'}}_L U \ket{\psi}_L=\sum_{i,j,k,l} b_{i,j,k,l} \scp{f_{i'}}{E_j} \scp{O_l}{\psi} \ketbra{E_i}{O_k} +C_{i,j,k,l} \scp{f_{i'}}{E_j} \scp{O_l}{\psi} \ketbra{O_i}{E_k}\\
\ket{\psi} \in \text{odd ,} \medspace \ket{f_{i'}} \in \text{odd :} \quad \qquad \bra{f_{i'}}_L U \ket{\psi}_L=\sum_{i,j,k,l} d_{i,j,k,l} \scp{f_{i'}}{O_j} \scp{O_l}{\psi} \ketbra{O_i}{O_k} +A_{i,j,k,l} \scp{f_{i'}}{O_j} \scp{O_l}{\psi} \ketbra{E_i}{E_k} 
\end{gather}

Thus we see that the $E_{i'}$ operators are linear SSR operators as established in Theorem \ref{thm:blockform}. Hence the implication is done. We observe that for the statement to hold is necessary to have the possibility of having anti-diagonal Kraus operators and that we can achieve the full generality of forms of the Kraus operators by choosing any $\ket{\psi}$.\\

Now lets proof the reverse, 1 implies 3, that follows from the immense amounts of degrees of freedom that one has to choose a unitary matrix.\\

Consider a map $\varphi$ from $\mathcal{H}_S^N$ to itself such that $\phi(\rho)= \sum_k E_k \rho E_k^\dagger$ for all SSR $\rho$, with $E_k$ a set of SSR linear operators such that $\sum_k E_k^\dagger E_k=\mathbb{I}_N$. We now consider a fermionic finite space generated by a number of modes $K$ equal to the number of Kraus operators $E_k$. We denote this new fermionic space by $L$, and by construction its dimension is $2^K$. We denote the elements of the canonical SSR basis $\mathcal{B}_L$ by $\{\ket{e_j}\}_{j=1}^{2^N}$. Once this is done we consider the following map:
\begin{gather}
V: \mathcal{H}_S^N \wedge \mathbb{C} \ket{e_1} \longrightarrow \mathcal{H}\wedge L \nonumber \\ \ket{\psi} \wedge \ket{e_1} \longmapsto \sum_{k} E_k \ket{\psi} \wedge \omega_k 
\end{gather}
where $\omega_k$'s choice is conditioned to how transforms the corresponding $E_k$ the parity. If $E_k$ is block diagonal then an even state is selected, and an odd state is selected if $E_k$ flips parities by being block anti-diagonal. Is important to point out that at each time a different element of the orthonormal basis $\mathcal{B}_L$ is chosen. This choice can always be made due to the fact that $K\leq 2^{K-1}$ always. This procedure ensures that $V$ preserves the parity sending even states to even states and odd states to odd states. This follows since $\ket{e_1}=\ket{\Omega}_L$, so it is an even state.\\

Since this map is an isometry, we can extend it to a unitary map from $\mathcal{H}_S^N\wedge L$ to itself. Since the restriction is weak, we have many degrees of freedom left. Given that the map preserves parity is not difficult to check that we can choose the extension to preserve parity by sending even states to even states and odd states to odd states. \\

So with this reasoning we obtain a unitary operator $U$ that acts on $\mathcal{H}_S^N\wedge L$ and that preserves parity. Now we are in conditions to claim that if we choose $\omega=\proj{e_1}=\proj{\Omega}$ the statement holds. It just has to be calculated: 
\begin{gather}
\tr_{L}\left(U\left(\rho\wedge\omega\right)U^\dagger\right)=\tr_{L}\left(U\left(\sum_i p_i \proj{\psi_i}\wedge\proj{e_1}\right)U^\dagger\right)=\sum_i p_i \tr_{E}\left(U\left(\ket{\psi_i}\wedge\ket{e_1}\right)\left(\bra{\psi_i}\wedge\bra{e_1}\right)U^\dagger\right)=\nonumber\\=\sum_i p_i \tr_{L}\left(V\left(\ket{\psi_i}\wedge\ket{e_1}\right)\left(\bra{\psi_i}\wedge\bra{e_1}\right)V^\dagger\right)=\sum_i p_i \tr_{L}\left(\left(\sum_k E_k\ket{\psi_i} \wedge \omega_k\right)\left(\sum_{k'} E_{k'}\ket{\psi_i}\wedge \omega_{k'}\right)^\dagger\right)=\nonumber\\=\sum_i\sum_k \sum_{k'}  p_i \tr_{L}\left(\left(E_k\proj{\psi_i}E_{k'}^\dagger \right)\wedge \left(\omega_k \omega_{k'}^\dagger \right)\right)=\sum_i\sum_k \sum_{k'}  p_i \left(E_k\proj{\psi_i}E_{k'}^\dagger \right) \left(\omega_{k'}^\dagger\omega_k  \right)= \nonumber \\=\sum_i\sum_k \sum_{k'}  p_i \left(E_k\proj{\psi_i}E_{k'}^\dagger \right) \delta_{k k'}=\sum_k E_k \left(\sum_i p_i \proj{\psi_i}\right) E_{k}^\dagger = \sum_k E_k \rho E_{k}^\dagger
\end{gather}
Just as desired. Thus, the implication holds.\\

Finally, the equivalence between statements 1 and 2 follows directly from Theorem \ref{thm:Kraus-CPTP}. Redoing the proof it can be seen easily that the trace-preserving property becomes equivalent to impose $\sum_k E_k^\dagger E_k = \mathbb{I}_N$. 
\end{proof}

\end{document}